\documentclass[10pt, Conference, letterpaper]{IEEEtran}
\IEEEoverridecommandlockouts
\usepackage{cite}
\usepackage[T1]{fontenc}
\usepackage{aecompl}
\usepackage[numbers,sort&compress]{natbib}
\usepackage{amsmath,amssymb,amsfonts}
\usepackage{amsthm}
\usepackage{float}
\usepackage{booktabs}
\makeatletter
\renewcommand{\maketag@@@}[1]{\hbox{\m@th\normalsize\normalfont#1}}
\makeatother
\usepackage{graphicx}
\usepackage{verbatim}
\usepackage{subcaption}
\usepackage{textcomp}
\usepackage{xcolor}
\usepackage{stfloats}
\usepackage{tabularx}
\usepackage{array}
\usepackage{makecell}
\usepackage{url}
\makeatother
\usepackage{textcomp}
\usepackage[marginal]{footmisc}
\usepackage{geometry}
\usepackage{setspace}
\allowdisplaybreaks[4]
\usepackage{algpseudocode}
 \usepackage[linesnumbered,ruled]{algorithm2e}

\def\BibTeX{{\rm B\kern-.05em{\sc i\kern-.025em b}\kern-.08em
    T\kern-.1667em\lower.7ex\hbox{E}\kern-.125emX}}
\newtheorem{myDef}{Definition}
\newtheorem{myTheo}{Theorem}

\newtheorem{myCoro}{Corollary}
\begin{document}
\title{NOMA-Assisted Multi-BS MEC Networks for Delay-Sensitive and Computation-Intensive IoT Applications}
\author{Yuang Chen, \IEEEmembership{Graduate Student Member}, Fengqian Guo, Chang Wu, Mingyu Peng, \IEEEmembership{Graduate Student Member}, \\ Hancheng Lu, \IEEEmembership{Senior Member, IEEE}, and Chang Wen Chen, \IEEEmembership{Life Fellow, IEEE}
\thanks{This work was supported in part by the National Science Foundation of China under Grant U21A20452, in part by Hong Kong Research Grants Council (GRF-15213322, GRF-15229423). Yuang Chen, Fengqian Guo, Chang Wu, Mingyu Peng, and Hancheng Lu are with the Laboratory of Future Networks, University of Science and Technology of China (USTC), Hefei, P. R. China (e-mail: \{yuangchen21, changwu, mypeng\}@mail.ustc.edu.cn; fqguo@ustc.edu.cn; hclu@ustc.edu.cn). Yuang Chen and Hancheng Lu are also with Deep Space Exploration Laboratory, Hefei 230088, China. Chang Wen Chen is with the Department of Computing, The Hong Kong Polytechnic University, Hong Kong, China (e-mail: changwen.chen@polyu.edu.hk).}}
\maketitle
\begin{abstract}
The burgeoning and ubiquitous deployment of the Internet of Things (IoT) landscape struggles with ultra-low latency demands for computation-intensive tasks in massive connectivity scenarios. In this paper, we propose an innovative uplink non-orthogonal multiple access (NOMA)-assisted multi-base station (BS) mobile edge computing (BS-MEC) network tailored for massive IoT connectivity. To fulfill the quality-of-service (QoS) requirements of delay-sensitive and computation-intensive IoT applications, we formulate a joint task offloading, user grouping, and power allocation optimization problem with the overarching objective of minimizing the system's total delay, aiming to address issues of unbalanced subchannel access, inter-group interference, computational load disparities, and device heterogeneity. To effectively tackle this problem, we first reformulate task offloading and user grouping into a non-cooperative game model and propose an exact potential game-based joint decision-making (EPG-JDM) algorithm, which dynamically selects optimal task offloading and subchannel access decisions for each IoT device based on its channel conditions, thereby achieving the Nash Equilibrium. Then, we propose a majorization-minimization (MM)-based power allocation algorithm, which transforms the original subproblem into a tractable convex optimization paradigm. Extensive simulation experiments demonstrate that our proposed EPG-JDM algorithm significantly outperforms state-of-the-art decision-making algorithms and classic heuristic algorithms, yielding performance improvements of up to 19.3\% and 14.7\% in terms of total delay and power consumption, respectively.
\end{abstract}
\begin{IEEEkeywords}
Non-orthogonal Multiple Access (NOMA); Internet of Things (IoT); Task Offloading; Game Theory; Mobile Edge Computing (MEC).
\end{IEEEkeywords}

\section{Introduction}

\IEEEPARstart{W}{ith} the ubiquitous deployment of fifth-generation (5G) mobile communication technology worldwide, the Internet-of-Things (IoT) has entered a stage of large-scale commercialization, encompassing emerging applications such as intelligent housing, wearable devices, urban infrastructure, and smartphones \cite{10879488, chowdhury2022survey, el2019survey, 9319633, 7902207}. Remarkably, over 127 new IoT devices are connected to the Internet every second, and it is estimated that by 2030, the number of connected IoT devices will reach to nearly 125 billion \cite{9319633}. According to the forecast by the International Data Company (IDC), the amount of data generated by IoT devices has exceeded $80$ZB in 2025 \cite{ARO2023_IoTData}. However, identifying and orchestrating such a massive number of heterogeneous IoT devices has become an increasingly challenging task \cite{7902207, chen2025aoiaware, 11037391}. The rapid proliferation of active IoT terminal devices globally has resulted in an exponential increase in data traffic. The generation of such massive data volumes introduces a surge in computing tasks, which are typically delay-sensitive and computation-intensive (DSCI) characteristics in real-world applications \cite{9667454, sun2024collaborative, 10517901}. However, due to constraints in hardware and battery capacity, most IoT devices pose very limited computing power, making it difficult to effectively meet the demands of users for DSCI tasks \cite{chen2025aoiaware, chowdhury2022survey}. As a result, ensuring the smooth execution of IoT networks with DSCI tasks has become an urgent issue to be addressed.

\vspace{-0.2em}

\par Generally, IoT devices necessitate wirelessly offloading the DSCI tasks requested by users to resource-rich servers for further processing \cite{10460318}. Traditional cloud computing paradigms address this by offloading DSCI tasks to the cloud center \cite{chen2025aoiaware, chen2025dmsa, 10876766}. However, due to the dispersed distribution of massive IoT access, centralized computing paradigms suffer from low transmission efficiency \cite{10460318}. Moreover, the centralized deployment of DSCI tasks in cloud centers leads to significantly increased cloud workloads, which not only reduces system resource utilization but also results in higher network delays. In this context, mobile edge computing (MEC) paradigms have gained widespread adoption, aiming to migrate computing and caching resources from remote cloud data centers to the network edge, closer to users and data sources \cite{chen2025aoiaware, chen2025dmsa, 10876766}. This physical spatial proximity adaptation has significant advantages over geographically dispersed IoT devices in reducing transmission delay, enhancing network reliability, and reducing network bandwidth pressure. Although offloading DSCI tasks to MEC servers can effectively enhance quality-of-service (QoS) provisioning performance, the surge of large-scale IoT access coupled with high-bandwidth traffic generated by applications like Internet-of-Video Things (IoVT), presents a formidable challenge \cite{10430407, 9129776}. These services now constitute a significant portion of network traffic and impose stringent requirements on delay and reliability \cite{9129776, 10430407}.

\vspace{-1em}

\subsection{Research Motivations and Challenges}

\par The development of IoT-enabled MEC networks with DSCI tasks hinges on tailored effective edge offloading and wireless transmission schemes. While extensive studies have established a solid foundation for facilitating massive IoT connectivity, some significant practical challenges remain, as follows:

\par (1) \textbf{Existing multi-access technologies typically neglect the interdependency between task offloading and user grouping.} As IoT devices with DSCI tasks connecting to MEC networks, the uplink transmission of high-volume DSCI tasks to servers inevitably encounters spectrum resource bottlenecks \cite{9760218,8606230}. Traditional orthogonal multiple access (OMA) technologies struggle to meet the offloading and delay requirements of DSCI tasks in resource-constrained networks. Non-orthogonal multiple access (NOMA) has presented the potential to enhance system spectral efficiency (SE) and gained significant attention by employing superposition encoding to combine signals to access the same resource block in the power domain or code domain \cite{10355071,10908860,10430407}. At the receiver, NOMA utilizes successive interference cancellation (SIC) to decode the superimposed signal for efficient multi-device communications \cite{10355071,10908860,10430407}.

\par However, existing studies have numerous limitations. On the one hand, the user grouping schemes of most NOMA systems are based on their channel qualities, where the number of users accessing each sub-channel is typically considered as identical \cite{10355071, 10214196, 10050140, 9393794}. Such user grouping strategies fail to adequately consider the interdependency between task offloading and user grouping, as well as the impacts of interference, power consumption, and delay requirements on strategy selection. On the other hand, NOMA systems relying on SIC technology may cause serious inter-device interference when transmitting data from different devices simultaneously \cite{10355071,10908860,10430407,10214196, 10050140}. Moreover, the user grouping schemes in existing NOMA systems rarely account for the heterogeneity of IoT devices, such as different delay demands for DSCI tasks \cite{10355071,10908860,10430407,10214196, 10050140}.

\par (2) \textbf{IoT access with DSCI tasks in multi-BS-MEC systems significantly exacerbates the complexity of task offloading and resource allocation.} The MEC networks supporting IoT connectivity with DSCI tasks generate huge volumes of transmission tasks and offloaded data \cite{chen2025aoiaware, 9667454}, imposing heavy pressure on communication and computation resources \cite{10361531, 8123913, 9760218}. Due to the limited transmission bandwidth and computing capacity, single-BS architectures are usually unable to fulfill the QoS requirements of DSCI tasks \cite{chen2025aoiaware}, which are also closely related to the number of subchannel resources and devices \cite{9393794}. As a result, single-BS architectures cannot flexibly handle diverse task offloading demands, leading to poor resource utilization and system performance.

\par Although the multi-BS architectures promise to alleviate single-point congestion by distributing the computational load, they introduce new challenges in dynamic wireless networks. On the one hand, the complexity of task offloading and user grouping becomes significantly aggravated. The multi-BS-MEC networks not only need to determine the optimal user grouping decision among multiple sub-channels for uplink transmission, but also select the optimal offloading target MEC server \cite{9393794, 10430407,chen2025aoiaware}, leading to intractable coupling and complexity among user grouping, task offloading, and resource allocation problems. On the other hand, inter-BS interference and coordination control issues are also prominent under NOMA mechanisms \cite{10430407,11089502}, as spectrum resource sharing among multiple BSs easily causes severe co-channel interference. Moreover, the lack of effective coordinated scheduling among multiple BSs leads to reduced resource utilization and uneven delays. Finally, performance optimization in multi-BS-MEC systems typically involves multi-dimensional variables, including BS selection, power allocation, task offloading, and user grouping, which result in an extensive and computationally complex search space for optimal solutions.

\vspace{-0.8em}

\subsection{Main Contributions}

\par In order to effectively overcome the aforementioned challenges, this paper proposes an uplink NOMA-assisted multi-BS-MEC (NMBM) network designed to facilitate the DSCI IoT applications. The proposed NMBM network comprehensively accounts for uneven device access across sub-channels, inter-group interference, imbalanced computational loads, and heterogeneous delay constraints among devices. We formulate a joint optimization problem that integrates task offloading, user grouping, and power allocation, with the objective of minimizing the total delay of the proposed network. To efficiently tackle this NP-hard problem that involves the coupling between discrete task offloading and user grouping and continuous power allocation decision variables, we first adopt a variable decomposition strategy to divide the original problem into two subproblems: (1) the joint task offloading and user grouping subproblem, and (2) the power allocation subproblem. Corresponding optimization algorithms are then developed to solve each subproblem efficiently. The primary contributions of this work are summarized as follows:

\begin{itemize}
  \item To support the IoT applications with DSCI tasks access to MEC networks, we propose an uplink NMBM network, which considers user grouping for IoT devices with non-uniform access on sub-channels, and takes into account inter-group interference, imbalanced computational load, and heterogeneous delay among devices. We formulate a joint optimization problem that accounts for task offloading, user grouping, and power allocation, with the goal of minimizing the system's total delay.

  \item For the subproblem of joint task offloading and user grouping, inspired by game theory, we reformulate it into a non-cooperative game model and propose an exact potential game-based joint decision-making (EPG-JDM) algorithm. Under the proposed game framework, users dynamically select the optimal task offloading and sub-channel access decisions (i.e., user grouping decisions) based on their own channel conditions, taking into account the strategies already determined by other users.

  \item For the power allocation subproblem, we leverage the Majorization-Minimization (MM) algorithm to transform the power allocation problem into a tractable convex optimization problem. Furthermore, we propose a joint alternating optimization algorithm to iteratively address these two subproblems until the objective function converges.

  \item Extensive simulations demonstrate that the proposed algorithms significantly outperform the most representative existing baseline algorithms including Max-Min-based Grouping \cite{10093902}, Gale-Shapley Grouping \cite{7972929}, Nearby BS-based Offloading \cite{10309191}, and Computing Capacity-based Offloading, and our proposed algorithms can achieve at least 19.3\% and 14.7\% performance improvements in terms of total delay and power consumption, respectively.
\end{itemize}

\par The remainder of this paper is organized as follows. Sec. II reviews related works. Sec. III introduces the proposed system model. Sec. IV presents the problem formulation and analysis. Sec. V provides the algorithm designs and solutions. Sec. VI provides extensive performance evaluations against the state-of-the-art comparison schemes. Finally, Sec. VII concludes the paper and explores future directions.

\section{Related Works}

\par This section first reviews the current research status of NOMA-assisted MEC networks under IoT access with DSCI tasks \cite{10321692, 10210075, 10251413, 10128149, 9129776, 10430407, 10336741, 10470427}. Then, we provide a comprehensive survey of task offloading and resource allocation in multi-BS-MEC networks \cite{10470427, 10221786, chen2025aoiaware, wei2025neurdora, 11095639}. Finally, we discuss key limitations of these related works that are closely aligned with our aforementioned challenges.

\vspace{-1.5em}

\subsection{Research Status of NOMA-assisted MEC Networks Under DCSI IoT Applications}

\par Benefiting from the advantages of NOMA in enhancing SE and transmission throughput, NOMA-assisted MEC networks that facilitate the IoT access with DSCI tasks have garnered significant attention \cite{10321692, 10210075, 10251413, 10128149, 9129776, 10430407, 10336741, 10470427}. The authors in \cite{10321692} proposed a two-way relaying NOMA-assisted IoT network, where the issue of the limited lifetime of IoT access points due to device energy constraints was thoroughly investigated. In \cite{10210075}, the authors introduced a jammer-aided covert reconfigurable intelligent surface (RIS)-NOMA system that exploits the NOMA characteristics to conceal the presence of strong signal transmissions to IoT receivers. To flexibly and accurately identify the usage of target frequencies, the authors in \cite{10251413} investigated an adaptive NOMA-based spectrum sensing approach for uplink IoT access networks. In \cite{10128149}, RIS technology was introduced to enhance the offloading capability of MEC systems. To further improve task offloading capability, the authors proposed an RIS-NOMA-assisted MEC network that jointly optimizes channel allocation, beam-bandwidth allocation, offloading rates, and power control schemes. IoVT imposes higher demands on transmission and computing capabilities of wireless networks than traditional IoT systems \cite{9129776}. In \cite{10430407}, we proposed a NOMA-assisted IoVT system combined with MEC and formulated a joint optimization problem for NOMA operations and MEC offloading to minimize the weighted average total system delay. In \cite{10336741}, the authors highlighted that existing downlink NOMA systems lead to redundant transmissions in resource-constrained satellite IoT networks; thus, they introduced the age-of-information (AoI) metric and proposed a content-aware sampling strategy to assess both the freshness and value of status updates. The authors in \cite{10470427} explored the dynamic task offloading problem in NOMA-based IoT networks, which integrates task scheduling and computing resource allocation decisions to support massive IoT connectivity.

\vspace{-1em}

\subsection{Survey on Transmission and Computing Schemes in Multi-BS-MEC Architectures}

\par IoT access with DSCI tasks led to a massive influx of data requiring both transmission and computational resources. To provide more powerful transmission and computation capacities, the multi-BS-MEC architectures have been extensively investigated \cite{10470427, 10221786, chen2025aoiaware, wei2025neurdora, 11095639}. In \cite{10470427}, the authors investigated a NOMA-assisted MEC system, where a dynamic task offloading and resource allocation joint optimization problem was formulated to minimize the system's energy consumption. The authors in \cite{10221786} proposed a RIS-assisted multi-MEC system, which can redirect the IoT devices to less-loaded MEC servers through passive beamforming techniques implemented on RISs when MEC servers are overwhelmed by massive computational tasks. In \cite{chen2025aoiaware}, we proposed an AoI-aware multi-BS-MEC real-time monitoring system to support large-scale industrial IoT access, where a joint task offloading and resource allocation optimization problem was formulated to ensure the information freshness of data packets. In \cite{wei2025neurdora}, the authors proposed a decentralized offloading algorithm for massive IoT-accessed multi-BS scenarios to ensure the effective and fair distribution of BS resources. The authors in \cite{11095639} developed a multi-BS collaborative edge computing model with heterogeneous computing capabilities, where a distributed training scenario was introduced to leverage the federated learning framework for collaborative optimization across multiple BSs while ensuring data privacy.

\vspace{-0.5em}

\subsection{Key Limitations of These Research Works}

\par Although the aforementioned studies on NOMA-assisted multi-MEC networks designed for DSCI IoT access have made significant progress in improving SE, task offloading, and resource allocation \cite{10321692, 10210075, 10251413, 10128149, 9129776, 10430407, 10336741, 10470427}, some key limitations remain. Specifically, these works typically overlook the close relationship between task offloading and user grouping, solely relying on channel quality or fixed device counts for user grouping, and fail to fully consider device heterogeneity (e.g., deadline constraints). This leads to increased interference among devices and makes it difficult to meet the QoS requirements of DSCI tasks. Furthermore, under the multi-BS-MEC architecture \cite{10470427, 10221786, chen2025aoiaware, wei2025neurdora, 11095639}, although distributed optimization \cite{wei2025neurdora} and federated learning \cite{11095639} have been introduced to distribute the load, the system still faces challenges such as strong coupling decision complexity, inter-BS co-channel interference, and insufficient coordinated scheduling, leading to poor resource utilization and higher delay in dynamic wireless environments. These limitations highlight the need for more refined joint optimization frameworks to address the technical challenges posed by DSCI IoT access.

\begin{figure*}[t]
\centering
\includegraphics[scale=0.95]{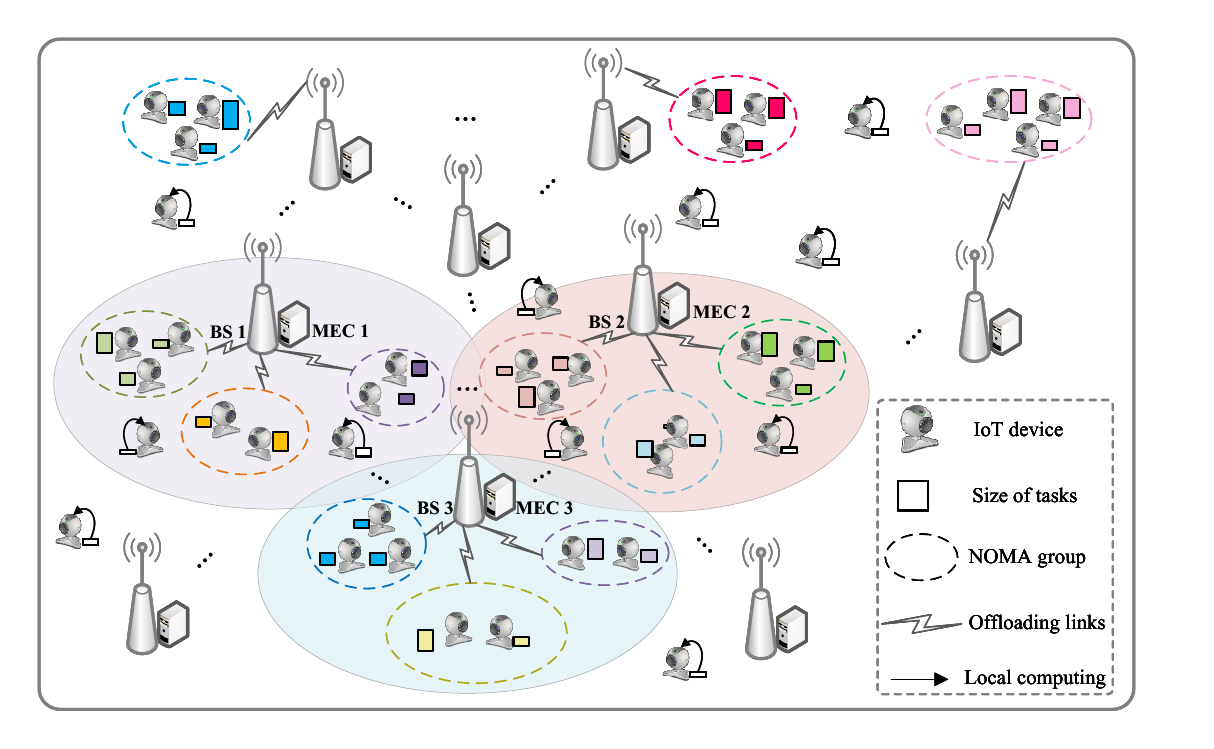}
\caption{\small The proposed Uplink NOMA-assisted multi-BS-MEC network for IoT access with DSCI tasks.}
\vspace{-1.5em}
\label{fig_system_model}
\end{figure*}

\section{System Model}

\vspace{-0.5em}

\par As illustrated in Fig. \ref{fig_system_model}, we propose an uplink NMBM network for IoT applications with DSCI tasks, which consists of $M$ BSs and $N$ single-antenna equipped IoT devices. IoT devices upload their collected sensor data to the MEC servers for processing through uplink wireless transmission. There is an MEC server deployed near each BS to provide sufficient computing resources, and this combination of BSs and MECs is referred to as BS-MEC. Each BS can occupy $G$ subchannels, and orthogonal spectrum resources are allocated among BSs to avoid inter-BS interference \cite{10470427,wei2025neurdora}. The sets of BS-MECs, IoT devices, and sub-channels can be denoted as $m \in \mathcal{M} \triangleq \left\{1, 2, \cdots, M\right\}$, $n \in \mathcal{N} \triangleq \left\{1, 2, \cdots, N\right\}$, and $g \in \mathcal{G} \triangleq \left\{1, 2, \cdots, G\right\}$, respectively. We use $\|\mathcal{U}_{mg}\|$ to denote the set of IoT devices that offload their DSCI tasks to BS-MEC $m$ through subchannel $g$. If the IoT device $n$ offloads its DSCI tasks to BS-MEC $m$ via subchannel $g$, then the channel coefficient between the IoT device $n$ and BS-MEC $m$ can be represented as $h_{n}^{mg}$. For all IoT devices that connect to the identical BS-MEC and jointly occupy the same subchannel, we leverage the NOMA technique to process their data transmission. As shown in Fig. \ref{fig_system_model}, the height of the rectangle next to each IoT device is used to intuitively reflect the amount of their computational demands, also referred to as the amount of data they require to process. Notably, IoT devices in the proposed system model are heterogeneous, which means their requirements for computing resources and delay are different.

\par Under NOMA mechanisms, the system throughput is not only associated with the channel conditions of individual IoT devices but also significantly influenced by the channel conditions of other devices within the same group. At the receiver end, the decoding process must follow a certain priority order, where the signals of IoT devices with the best channel conditions are decoded first, while the signals of other IoT devices are treated as interference. Based on this principle, decoding can be carried out step by step until the signals from all IoT devices within the group are successfully decoded. This decoding strategy ensures that the proposed network can fully utilize devices with better channel conditions to enhance the overall throughput. For the IoT device $n$ that selects to offload its DSCI tasks to the MEC $m$ via subchannel $g$, the throughput of IoT device $n$ can be expressed as

\vspace{-0.4em}

\begin{equation}\label{eq1}
    R_{n}^{mg} = B \log_{2}\left\{1 + \frac{p_{n} |h_{n}^{mg}|^{2}}{\sum\limits_{k \in \mathcal{U}_{mg}, k\neq n, |h_{k}^{mg}| < |h_{n}^{mg}|} p_{k}|h_{k}^{mg}|^{2} + \sigma^{2}} \right\},
\end{equation}
where $p_{k}$ denotes the transmit power of the IoT device $k \in \mathcal{N}$, $\sigma^{2}$ indicates the noise power, and $B$ represents the bandwidth of the occupied subcarrier.

\par Each IoT device faces two choices for processing its DSCI tasks. One is to utilize its own computation capacity to process tasks locally, and the other is to offload its tasks to candidate MEC servers for processing. To characterize whether IoT device $n$ chooses to offload its tasks to BS-MEC $m$, we introduce an offloading decision variable $a_{n,m}$, which can be defined as follows:

\vspace{-0.5em}

\begin{equation}\label{eq2}
    a_{n,m} = \left\{
                \begin{array}{ll}
                  1, & \hbox{offload $n$'s tasks to BS-MEC $m$;} \\
                  0, & \hbox{process tasks locally.}
                \end{array}
              \right.
\end{equation}

\par When the IoT device $n$ processes its tasks locally, the required processing time can be denoted as

\begin{equation}\label{eq3}
     T_{n}^{\rm{local}} = \frac{W_{n}}{f_{n}^{\rm{local}}},
\end{equation}
where $W_{n}$ represents the computational demands, also referred to as the size of the task needed to be processed, and $f_{n}^{\rm{local}}$ indicates the computation capacity of the device $n$ itself.

\par When the IoT device $n$ selects offloading its DSCI tasks to the BS-MEC $m$ via the subchannel $g$, then the transmission delay of the device $n$ can be given by

\begin{equation}\label{eq4}
    T_{n,m,g}^{\rm{trans}} = \frac{W_{n}}{R_{n}^{mg}}.
\end{equation}

\par Since there may be multiple IoT devices offloading their DSCI tasks to the identical BS-MEC $m$, it inevitably leads to the issue of competition for computing resources. Therefore, the computing delay of IoT device $n$ can be expressed as

\begin{equation}\label{eq5}
    T_{n,m}^{\rm{comp}} = \frac{W_{n} + \sum\limits_{i \in \mathcal{N}, i \neq n} a_{i, m}W_{i} }{f_{m}},
\end{equation}
where $f_{m}$ denotes the computation capacity of BS-MEC $m$, and the term $\sum\limits_{i \in \mathcal{N}, i \neq n} a_{i, m}W_{i}$ indicates total data volume of all IoT devices that offload tasks to BS-MEC $m$ except for device $n$. Therefore, the total delay when device $n$ chooses MEC offloading can be denoted as

\begin{equation}\label{eq6}
   \begin{aligned}
        T_{n,m,g}^{\rm{mec}} & = T_{n,m,g}^{\rm{trans}} + T_{n,m}^{\rm{comp}} = \frac{W_{n}}{R_{n}^{mg}} + \frac{W_{n} + \sum\limits_{i \in \mathcal{N}, i \neq n} a_{i, m}W_{i} }{f_{m}}.
   \end{aligned}
\end{equation}

\par Then, the delay expressions $T_{n}$ of IoT device $n$ under these two offloading strategies can be given by

\begin{equation}\label{eq7}
    T_{n} = \left\{
                \begin{array}{ll}
                  \frac{W_{n}}{R_{n}^{mg}} + \frac{W_{n} + \sum\limits_{i \in \mathcal{N}, i \neq n} a_{i, m}W_{i} }{f_{m}}, & \hbox{BS-MEC Offloading;} \\
                  \frac{W_{n}}{f_{n}^{\rm{local}}}, & \hbox{Local Offloading.}
                \end{array}
              \right.
\end{equation}

\par Moreover, to describe the user grouping, we introduce the variables $b_{g,n}$ to indicate whether IoT device $n$ is connected to sub-channel $g$, and $b_{g,n}$ can be specifically defined as

\begin{equation}\label{eq8}
    b_{g,n} = \left\{
                \begin{array}{ll}
                  1, & \hbox{IoT device $n$ access group $g$;} \\
                  0, & \hbox{IoT device $n$ unconnect to group $g$.}
                \end{array}
              \right.
\end{equation}

\section{Problem Formulation and Analysis}

\par Building on the analysis in Sec. III, this section aims to jointly optimize the task offloading, user grouping, and power allocation schemes under the scenarios of device heterogeneity, to minimize the total delay of the proposed uplink NMBM networks. The detailed optimization problem is formulated as follows:

\vspace{-2em}

\begin{subequations}\label{eq9}
\begin{align}
\mathcal{P}1:     & \min_{\{\mathcal{U}_{m,g}, p_{n}\}_{\forall n \in \mathcal{N}, m \in \mathcal{M}, g \in \mathcal{G}}} \sum\limits_{n = 1}^{N} T_{n}, \\
\emph{s.t.} \quad \quad & 0 \leq p_{n} \leq p_{max}, \quad \forall n \in \mathcal{N}, \\
                        & 0 \leq T_{n} \leq D_{n}, \quad \forall n \in \mathcal{N},   \\
                        & \sum\limits_{m = 1}^{M} a_{n,m} \leq 1, \quad n \in \mathcal{N},\\
                        & \sum\limits_{g = 1}^{G} b_{g,n} \leq 1, \quad n \in \mathcal{N},\\
                        & \mathcal{U}_{m,g} \cap \mathcal{U}_{m,g^{\prime}} = \varnothing, \quad \forall g, g^{\prime} \in \mathcal{G}, g \neq g^{\prime},\\
                        & \mathcal{U}_{m,g} \cap \mathcal{U}_{m^{\prime},g} = \varnothing, \quad \forall m, m^{\prime} \in \mathcal{M}, m \neq m^{\prime},
\end{align}
\end{subequations}
where constraint (\ref{eq9}b) denotes the transmission power limit of the IoT device, in which $p_{max}$ indicates the maximum transmit power of each device. Constraint (\ref{eq9}c) represents the delay threshold of DSCI tasks, which means that each task must be completed within that time. Constraints (\ref{eq9}d) and (\ref{eq9}e) indicate that each device can only access a maximum of one BS-MEC and sub-channel during task offloading, respectively. Finally, constraints (\ref{eq9}f) and (\ref{eq9}g) represent that IoT devices accessing different groups and offloading to different BS-MECs do not overlap with each other.

\par The original problem $\mathcal{P}1$ involves a mix of discrete and continuous variables, and extremely complex interdependent coupling between task offloading, user grouping, and power allocation. Thus, $\mathcal{P}1$ is classified as a Mixed Integer Nonlinear Programming (MINLP) problem, and directly tackling $\mathcal{P}1$ presents significant challenges. To simplify the solution of $\mathcal{P}1$, we decompose the discrete and continuous variables for problem solution. We first fix the power allocation strategy $\boldsymbol{p} = \{p_{1}, \cdots, p_{N}\}$, and leverages game theory to design efficient algorithms. Based on the obtained task offloading and user grouping strategies, the power allocation strategy $\boldsymbol{p} = \{p_{1}, \cdots, p_{N}\}$ can be effectively solved, at which point the optimization objective in $\boldsymbol{T} = \{T_{1}, \cdots, T_{N}\}$ is solely related to $\boldsymbol{p} = \{p_{1}, \cdots, p_{N}\}$.

\vspace{-0.5em}

\section{Algorithm Designs and Solutions}

\par In this section, we propose the corresponding algorithms for the joint task offloading and user grouping subproblem, as well as the power allocation subproblem. Then, we develop a joint iterative optimization algorithm to effectively address $\mathcal{P}1$.

\vspace{-1em}

\subsection{Potential Game-Based Joint Task Offloading and User Grouping Scheme}

\par Based on the structure of the joint task offloading and user grouping subproblem, we intend to leverage game theory to tackle it. Let $\pi_{n}$ denote the device $n$'s task offloading and user grouping selection, where $\pi_{n} = \{m, g\} \Leftrightarrow n \in \mathcal{U}_{m,g}$. The strategy profile of all devices can be represented as $\boldsymbol{\pi} = \{\pi_{1}, \cdots, \pi_{N}\}$, and $\pi_{-n}$ denotes the strategies of all devices except $n$. According to the objective function $T_{n}$ formulated in (\ref{eq7}), if the device's local computing resources can satisfy the delay requirements, it computes locally without uplink transmission. Otherwise, it offloads its DSCI tasks to the BS-MEC $m$, requiring a suitable offloading and grouping design. Based on (\ref{eq6}), when the device $n$ offloads its tasks to BS-MEC $m$, the total delay is composed of transmission and computation delays, where the former part is impacted by the intra-group interference, and the second part is related to the computing resource competition from other devices offloaded on the same BS-MEC $m$. Specifically, greater intra-group interference reduces throughput and increases transmission delay. Moreover, when the competition for computing resources among devices intensifies, the computation delay will also increase accordingly. To model these dynamic interactions among IoT devices, we define an interference function as follows:

\begin{equation}\label{eq10}
    \nu_{n} = \sum\limits_{k \in \mathcal{U}_{m,g}, k \neq n, |h_{k}^{mg}| < |h_{n}^{mg}|} p_{k} \left|h_{k}^{mg}\right|^{2} + \sum\limits_{i \in \mathcal{N}, i \neq n} W_{i} a_{i,m},
\end{equation}
where the first term denotes the intra-group interference caused by other IoT devices that select the identical sub-channel for task transmission, and the second term reflects the competition for computing resources triggered by other devices that offload their tasks to the identical BS-MEC.

\par Based on the above analysis, we transform the objective of minimizing the total delay (\ref{eq9}a) into minimizing the sum of interference from all devices. Then, $\mathcal{P}1$ can be reconstructed into a non-cooperative game model, where each IoT device acts as an independent game player, characterized by \texttt{\textbf{absolute rationality}} and \texttt{\textbf{self-interest}}, aiming to minimize its own interference. In this case, the game problem can be expressed using the tuple $\Gamma = \left(\mathcal{N}, \mathcal{M}, \mathcal{G}, \{h_{n}\}_{n \in \mathcal{N}}, \{\nu_{n}\}_{n \in \mathcal{N}}\right)$, where $\mathcal{N}$ denotes the set of game players, $\nu_{n}$ indicates the defined interference function of device $n$. Notably, $\Gamma$ is a deterministic strategy game, which means that the strategies of players are explicit and do not involve probabilistic transitions \cite{wang2010game}. In game theory, when all players reach a state where no player can unilaterally change their strategy, we refer to this state as a Nash Equilibrium (NE):

\begin{myDef}\label{def1}
   \textbf{(Nash Equilibrium)} For the game player $\forall n \in \mathcal{N}$, the system reaches the NE state when the strategy $\boldsymbol{\pi}^{\ast} = \left(\pi_{1}^{\ast}, \pi_{2}^{\ast}, \cdots, \pi_{N}^{\ast}\right)$ satisfies the following conditions, as follows:
   \begin{equation}\label{eq11}
       \nu_{n}\left(n \in \mathcal{U}^{\pi_{n}^{\ast}},\boldsymbol{\pi}^{\ast}_{-n}\right) \leq \nu_{n}\left(n \in \mathcal{U}_{m,g},\boldsymbol{\pi}^{\ast}_{-n}\right),
   \end{equation}
   where $n \in \mathcal{U}^{\pi_{n}^{\ast}}$ represents that the task offloading and grouping scheme belongs to $\pi_{n}^{\ast}$.
\end{myDef}

\par The system only converges when the game reaches the NE state; however, not all games can achieve the NE state. To guarantee the system's convergence, we transform the original problem $\mathcal{P}1$ into an Exact Potential Game (EPG) model, which ensures the existence of an NE. The definition of the EPG can be given as follows:

\begin{myDef}\label{def2}
     \textbf{(Exact Potential Game)} For $\forall n \in \mathcal{N}$, $\forall g, g^{\prime} \in \mathcal{G}$, and $\forall m, m^{\prime} \in \mathcal{M}$, if there exists a function $\Phi_{n}$ that satisfies the condition (\ref{eq12}), the game $\Gamma = \left(\mathcal{N}, \mathcal{M}, \mathcal{G}, \{h_{n}\}_{n \in \mathcal{N}}, \{\nu_{n}\}_{n \in \mathcal{N}}\right)$ is a strictly EPG and the function $\Phi_{n}$ represents the exact potential function of the EPG $\Gamma$, as follows:
     \begin{equation}\label{eq12}
         \begin{aligned}
              & \Phi_{n}\left(n \in \mathcal{U}_{m,g}, \boldsymbol{\pi}^{\ast}_{-n} \right) - \Phi_{n}\left(n \in \mathcal{U}_{m^{\prime},g^{\prime}}, \boldsymbol{\pi}^{\ast}_{-n} \right) \\
              & = \nu_{n}\left(n \in \mathcal{U}_{m,g}, \boldsymbol{\pi}^{\ast}_{-n} \right) - \nu_{n}\left(n \in \mathcal{U}_{m^{\prime},g^{\prime}}, \boldsymbol{\pi}^{\ast}_{-n} \right).
         \end{aligned}
     \end{equation}
\end{myDef}

\par The \textbf{Definition \ref{def2}} reveals the change in potential function $\Phi_{n}$ during the transfer of device $n$ from one offloading-group combination $\left(m, g\right)$ to another offloading-group combination $\left(m^{\prime}, g^{\prime}\right)$, which is equivalent to the change in the original interference function $\nu_{n}$. Although a detailed definition of the EPG's function has been provided in \textbf{Definition \ref{def2}}, there has been no specific proof of whether the game $\Gamma$ satisfies the conditions of an EPG, which can be proven in \textbf{Theorem \ref{theo1}}, as follows:

\begin{myTheo}\label{theo1}
    Based on Definition \ref{def2}, the closed-form expression of the $\Phi_{n}$ is derived by (\ref{eq14}), which makes Eq. (\ref{eq12}) established.
\end{myTheo}

\begin{proof}
    According to the general deductive method of game theory \cite{wang2010game, 10181238}, we give the expression $\Phi_{n}$ that may satisfy (\ref{eq12}) based on (\ref{eq10}), as follows:

    \vspace{-0.5em}

    \begin{equation}\label{eq13}
        \begin{aligned}
             & \Phi_{n}\left(n \in \mathcal{U}_{m,g}, \boldsymbol{\pi}^{\ast}_{-n} \right) \\
             & = \frac{1}{2 W_{n}} \sum\limits_{u = 1}^{N} \sum\limits_{i = 1, i \neq u}^{N} \ell_{\{a_{i,m} = a_{u,m}\}} W_{i} W_{u} - \sum\limits_{i = 1, i \neq n} p_{i} |h_{i}^{mg}|^{2} \\
             & + 2 \sum\limits_{u = 1}^{N} \sum\limits_{\substack{i \in \mathcal{U}_{mg}, \\ |h_{i}^{mg}| < |h_{u}^{mg}|}} \ell_{\{a_{i,m} = a_{u,m}\}} \ell_{\{b_{g,i} = b_{g,u}\}}  p_{i} |h_{i}^{mg}|^{2},
        \end{aligned}
    \end{equation}
    where $\ell_{\rm{condition}}$ denotes a binary indicator that takes the value 1 when the indicated $\rm{condition}$ is true, and takes the value 0 otherwise. The term $\sum\limits_{u = 1}^{N} \sum\limits_{i = 1, i \neq u}^{N} \ell_{\{a_{i,m} = a_{u,m}\}} W_{i} W_{u}$ can be discussed in three cases as follows:

    \begin{enumerate}
        \item If $u = n, i \neq n, \Rightarrow \sum\limits_{i \neq n} \ell_{\{a_{i,m} = a_{n,m}\}} W_{i} W_{n}$;

        \item If $i = n, u \neq n \Rightarrow \sum\limits_{u \neq n} \ell_{\{a_{u,m}= a_{n,m\}}} W_{n} W_{u}$;

        \item If $i \neq n, u \neq n \Rightarrow \sum\limits_{u \neq n} \sum\limits_{\substack{i \neq u \\ i \neq n}} \ell_{\{a_{i,m} = a_{u,m}\}} W_{i} W_{u}$.
    \end{enumerate}

    \vspace{-0.2em}

    Therefore, we can derive that

    \vspace{-0.5em}

    \begin{equation}\label{eq14}
        \begin{aligned}
            & \quad \frac{1}{2 W_{n}} \sum\limits_{u = 1}^{N} \sum\limits_{i = 1, i \neq u}^{N} \ell_{\{a_{i,m} = a_{u,m}\}} W_{i} W_{u}\\
            & = \frac{1}{2 W_{n}} \bigg\{\sum\limits_{i \neq n} \ell_{\{a_{i,m} = a_{n,m}\}} W_{i} W_{n} +  \sum\limits_{u \neq n} \ell_{\{a_{u,m}= a_{n,m\}}} W_{n} W_{u}\bigg\} \\
            & \quad + \underbrace{\frac{1}{2 W_{n}} \sum\limits_{u \neq n} \ell_{\{a_{u,m}= a_{n,m\}}} W_{n} W_{u}}_{\mathcal{C}\left(n \in \mathcal{U}_{m,g}\right)} \\
            & = \sum\limits_{\substack{i = 1 \\ i \neq n}} \ell_{\{a_{i,m} = a_{n,m}\}} W_{i} + \mathcal{C}\left(n \in \mathcal{U}_{m,g}\right).
        \end{aligned}
    \end{equation}

    We define $H(u)$ as follows:

    \begin{equation}\label{eq15}
        \begin{aligned}
            H(u) = \sum\limits_{\substack{i \in \mathcal{U}_{mg}, \\ |h_{i}^{mg}| < |h_{u}^{mg}|}} \ell_{\{a_{i,m} = a_{u,m}\}} \ell_{\{b_{g,i} = b_{g,u}\}}  p_{i} |h_{i}^{mg}|^{2}.
        \end{aligned}
    \end{equation}

    Then, we can derive that

    \begin{equation}\label{eq16}
         \begin{aligned}
             & \quad \quad - \sum\limits_{i = 1, i \neq n} p_{i} |h_{i}^{mg}|^{2} + 2 \sum\limits_{u = 1}^{N} H(u) \\
             &  = - \sum\limits_{i = 1, i \neq n} p_{i} |h_{i}^{mg}|^{2} + 2H(n) + 2 \sum\limits_{u \neq n}^{N} H(u),
         \end{aligned}
    \end{equation}
    where we can further derive that

    \begin{equation}\label{eq17}
        \begin{aligned}
            & - \sum\limits_{i = 1, i \neq n} p_{i} |h_{i}^{mg}|^{2} + 2H(n) \\
            & = - \sum\limits_{i = 1, i \neq n} p_{i} |h_{i}^{mg}|^{2} + 2 \sum\limits_{\substack{i \in \mathcal{U}_{mg}, \\ |h_{i}^{mg}| < |h_{u}^{mg}|}}  p_{i} |h_{i}^{mg}|^{2} \\
             & = \sum\limits_{\substack{i \in \mathcal{U}_{mg}, \\ |h_{i}^{mg}| < |h_{u}^{mg}|}}  p_{i} |h_{i}^{mg}|^{2} - \sum\limits_{\substack{i \in \mathcal{U}_{m,g} \\ i \neq n}} p_{i} |h_{i}^{mg}|^{2} - \sum\limits_{\substack{i \notin \mathcal{U}_{m,g}}} p_{i} |h_{i}^{mg}|^{2}\\
             & \quad + \sum\limits_{\substack{i \in \mathcal{U}_{mg}, \\ |h_{i}^{mg}| < |h_{u}^{mg}|}}  p_{i} |h_{i}^{mg}|^{2} \\
             & = - \!\!\! \sum\limits_{\substack{i \in \mathcal{U}_{mg}, \\ |h_{i}^{mg}| > |h_{u}^{mg}|}}  \!\!\!\! p_{i} |h_{i}^{mg}|^{2}  - \!\! \sum\limits_{\substack{i \notin \mathcal{U}_{m,g}}} p_{i} |h_{i}^{mg}|^{2} + \!\!\!\!\!\! \sum\limits_{\substack{i \in \mathcal{U}_{mg}, \\ |h_{i}^{mg}| < |h_{u}^{mg}|}}  \!\!\!\!\!\!\! p_{i} |h_{i}^{mg}|^{2}. \\
        \end{aligned}
    \end{equation}

    \vspace{-1em}

    Based on (\ref{eq16}) and (\ref{eq17}), we can obtain that

    \vspace{-1.5em}

    \begin{equation}\label{eq18}
        \begin{aligned}
            & \quad \quad - \sum\limits_{i = 1, i \neq n} p_{i} |h_{i}^{mg}|^{2} + 2 \sum\limits_{u = 1}^{N} H(u) \\
            & = \biggl\{ 2 \sum\limits_{\substack{j = 1 \\ j \neq n}}\ \ \sum\limits_{\substack{i \in \mathcal{U}_{mg}, \\ |h_{i}^{mg}| < |h_{u}^{mg}| \\ i \neq n, i \neq j}}  p_{i} |h_{i}^{mg}|^{2} - \sum\limits_{\substack{i \in \mathcal{U}_{mg}, \\ |h_{i}^{mg}| > |h_{u}^{mg}|}}  \!\!\!\! p_{i} |h_{i}^{mg}|^{2}  \\
            & \quad - \sum\limits_{\substack{i \notin \mathcal{U}_{m,g}}} p_{i} |h_{i}^{mg}|^{2} \biggl\} + \sum\limits_{\substack{i \in \mathcal{U}_{mg}, \\ |h_{i}^{mg}| < |h_{u}^{mg}|}}  p_{i} |h_{i}^{mg}|^{2}\\
            & = \mathcal{T}(n \in \mathcal{U}_{m,g}) + \sum\limits_{\substack{i \in \mathcal{U}_{mg}, \\ |h_{i}^{mg}| < |h_{u}^{mg}|}}  p_{i} |h_{i}^{mg}|^{2}.
        \end{aligned}
    \end{equation}

     \vspace{-0.5em}

    Combining (\ref{eq13}), (\ref{eq14}), and (\ref{eq18}), the potential function $\Phi_{n}\left(n \in \mathcal{U}_{m,g}, \boldsymbol{\pi}^{\ast}_{-n} \right)$ can be rewritten as follows:

     \vspace{-1.5em}

    \begin{equation}\label{eq19}
        \begin{aligned}
            \Phi_{n}\left(n \in \mathcal{U}_{m,g}, \boldsymbol{\pi}^{\ast}_{-n} \right) & = \sum\limits_{\substack{i = 1 \\ i \neq n}} \ell_{\{a_{i,m} = a_{n,m}\}} W_{i} + \!\!\!\!\!\!\!\! \sum\limits_{\substack{i \in \mathcal{U}_{mg}, \\ |h_{i}^{mg}| < |h_{u}^{mg}|}}  \!\!\!\!\! p_{i} |h_{i}^{mg}|^{2}\\
            & + \mathcal{C}\left(n \in \mathcal{U}_{m,g}\right) + \mathcal{T}(n \in \mathcal{U}_{m,g}),
        \end{aligned}
    \end{equation}
    where the expression of $\mathcal{C}\left(n \in \mathcal{U}_{m,g}\right)$ and $\mathcal{T}\left(n \in \mathcal{U}_{m,g}\right)$ are independent of the offloading and grouping strategies $\left(m,g\right)$ of device $n$, and they can be regarded as constant terms. Therefore, when the offloading and grouping strategies of IoT device $n$ change from $\mathcal{U}_{m,g}$ to $\mathcal{U}_{m^{\prime}, g^{\prime}}$, the corresponding change in the value of the EPF $\Phi_{n}$ can be denoted as follows:

    \vspace{-1em}

    \begin{subequations}\label{eq20}
        \begin{align}
            & \Phi_{n}\left(n \in \mathcal{U}_{m,g}, \boldsymbol{\pi}^{\ast}_{-n} \right) - \Phi_{n}\left(n \in \mathcal{U}_{m^{\prime},g^{\prime}}, \boldsymbol{\pi}^{\ast}_{-n} \right) \nonumber \\
            & = \bigg(\sum\limits_{\substack{i = 1 \\ i \neq n}} \ell_{\{a_{i,m} = a_{n,m}\}} W_{i} + \sum\limits_{\substack{i \in \mathcal{U}_{m,g} \\ |h^{mg}_{i}| < |h^{mg}_{n}|}} p_{i} |h^{mg}_{i}|^{2} \bigg) \nonumber \\
            & - \bigg(\sum\limits_{\substack{i = 1 \\ i \neq n}} \ell_{\{a_{i, m^{\prime}} = a_{n,m^{\prime}}\}} W_{i} + \!\!\!\!\!\! \sum\limits_{\substack{i \in \mathcal{U}_{m^{\prime},g^{\prime}} \\ |h^{m^{\prime}g^{\prime}}_{i}| < |h^{m^{\prime} g^{\prime}}_{n}|}} p_{i} |h^{m^{\prime}g^{\prime}}_{i}|^{2} \bigg) \nonumber\\
            & = \sum\limits_{\substack{i = 1 \\ i \neq n}} \ell_{\{a_{i,m} = a_{n,m}\}} W_{i} + \sum\limits_{\substack{i \in \mathcal{U}_{m,g} \\ |h^{mg}_{i}| < |h^{mg}_{n}|}} p_{i} |h^{mg}_{i}|^{2} \nonumber \\
            & - \sum\limits_{\substack{i = 1 \\ i \neq n}} \ell_{\{a_{i,m^{\prime}} = a_{n,m^{\prime}}\}} W_{i} - \!\!\!\!\! \sum\limits_{\substack{i \in \mathcal{U}_{m^{\prime},g^{\prime}} \\ |h^{m^{\prime}g^{\prime}}_{i}| < |h^{m^{\prime} g^{\prime}}_{n}|}} p_{i} |h^{m^{\prime} g^{\prime}}_{i}|^{2} \nonumber\\
            & = \nu_{n}\left(n \in \mathcal{U}_{m,g}, \boldsymbol{\pi}^{\ast}_{-n} \right) - \nu_{n}\left(n \in \mathcal{U}_{m^{\prime},g^{\prime}}, \boldsymbol{\pi}^{\ast}_{-n} \right),
        \end{align}
    \end{subequations}
    which satisfies (\ref{eq12}) in the \textbf{Definition \ref{def2}}.
\end{proof}

\par Based on \textbf{Definition \ref{def2}} and (\ref{eq10}), the change in the interference function (\ref{eq10}) has the same trend with the EPF, thus we can claim the existence of the NE state according to \textbf{Theorem \ref{theo1}}, as described in \textbf{Corollary \ref{coro1}} as follows:

\begin{myCoro}\label{coro1}
    The game $\Gamma = \left(\mathcal{N}, \mathcal{M}, \mathcal{G}, \{h_{n}\}_{n \in \mathcal{N}}, \{\nu_{n}\}_{n \in \mathcal{N}}\right)$ can converge to at least one NE state.
\end{myCoro}
\begin{proof}
    As defined in \textbf{Definition \ref{def1}}, a NE denotes is a state, where each IoT device selects its optimal task offloading and user grouping strategy given the strategies of others, such that no active device can reduce its experienced inter-group interference by unilaterally changing its strategy. As defined in \textbf{Definition \ref{def2}}, every such game has at least one NE, implying that the game is destined to reach a NE within a finite number of steps. In our considered game $\Gamma = \left(\mathcal{N}, \mathcal{M}, \mathcal{G}, \{h_{n}\}_{n \in \mathcal{N}}, \{\nu_{n}\}_{n \in \mathcal{N}}\right)$, the set of device $\mathcal{N}$ is finite, and the strategy evolution towards minimizing the value of interference function. Once the desired objective is achieved, strategies remain unchanged and are repeatedly selected, indicating that a NE has been reached.
\end{proof}

\vspace{-0.8em}

\par As mentioned above, in the potential game model, each player adjusts their strategy towards reducing their own interference. This behavior pattern follows the Finite Improvement Property (FIP) that is defined as follows:

\vspace{-0.3em}

\begin{myDef}\label{def3}
    \textbf{Finite Improvement Property (FIP)} In potential games, due to the fixed number of strategy spaces, the increasing path length of a game is limited. Therefore, after a limited number of strategy iterations and adjustments, players in a game reach an NE state, enabling the system to achieve optimal performance. This process is known as FIP, which gradually improves system performance through a finite number of iterations and ultimately reaches NE, ensuring the convergence of the system in games involving absolutely rational and selfish players.
\end{myDef}

\vspace{-0.4em}

\par According to \textbf{Definition \ref{def3}}, \textbf{Corollary \ref{coro2}} is derived as follows:

\begin{myCoro}\label{coro2}
    $\Gamma = \left(\mathcal{N}, \mathcal{M}, \mathcal{G}, \{h_{n}\}_{n \in \mathcal{N}}, \{\nu_{n}\}_{n \in \mathcal{N}}\right)$ meets the conditions of the FIP.
\end{myCoro}

\vspace{-1em}

\begin{proof}
    In the considered $\! \Gamma \!=\! \left(\mathcal{N}, \mathcal{M}, \mathcal{G}, \{h_{n}\}_{n \in \mathcal{N}}, \{\nu_{n}\}_{n \in \mathcal{N}}\right)$, due to the finite number of IoT devices, BSs, and subchannels, the corresponding sets of devices $\mathcal{N}$, BSs $\mathcal{M}$, and subchannels $\mathcal{G}$ are also limited. This means that the number of optional strategy spaces for offloading and grouping in $\Gamma$ is fixed. In addition, when players in $\Gamma$ adjust their strategies, their strategies evolve towards reducing interference and improving individual interests, ensuring that they do not repeatedly choose the same offloading and grouping strategies $(m,g)$. Therefore, $\Gamma$ will stop at an optimal strategy point, reaching a NE state after a limited number of iterative adjustments.
\end{proof}

\vspace{-0.5em}

\par In order to guarantee that the game results can achieve optimal performance, it is also necessary to ensure that the game $\Gamma$ is Pareto optimal, which can be defined as follows:

\begin{myDef}\label{def4}
    \textbf{Pareto Optimality (PO)} describes an ideal state of offloading and grouping strategies in which no players' situation can be worsened without improving the utility of at least one other player, meaning that the system can benefit one or more players without harming any others. A state is considered Pareto optimal when it is impossible to make any further improvement to the utility of at least one player without negatively impacting the others. In other words, a Pareto optimal state is one where no changes in offloading and grouping strategies can increase the performance of any player without causing harm to others, indicating that the system has reached a global optimal solution. Pareto improvement refers to the process and method of moving from an initial state to a Pareto optimal state.
\end{myDef}

\vspace{-0.5em}

\begin{algorithm}[h]
\small
 \setstretch{1.0}
        \caption{The Potential Game-based Optimal Joint Offloading and Grouping Algorithm.}
        \label{algo1}
        \KwIn{IoT device set $\mathcal{N}$, BS set $\mathcal{M}$, subchannel set $\mathcal{G}$, channel coefficient $h^{m g}_{n}$, delay constraints $D_{n}$;}
        \KwOut{Optimal task offloading and grouping strategies $\boldsymbol{\pi}^{\ast}$.}
        \textbf{Initialize:} $\mathcal{U}_{m,g} = \varnothing, \forall m \in \mathcal{M}, g \in \mathcal{G}$;\\
        Randomly initialize the offloading and grouping strategy $\pi_{n}$ for each device $n \in \mathcal{N}$;\\
        \Repeat{$\boldsymbol{\pi}^{\ast} = \{\pi_{n}^{\ast}\}_{n \in \mathcal{N}} = (m_{n}^{\ast}, g_{n}^{\ast})_{n \in  \mathcal{N}}$ remains unchanged}{
         \For{$t = 1$ to $T$}{
           \For{$n = 1$ to $N$}{
              Calculate $T_{n}^{\rm{local}}$ based on Eq. (\ref{eq3});\\
             \eIf{$T_{n}^{\rm{local}} \leq D_{n}$}
             {
                 \For{$j = 1$ to $M$}{
                   $a_{n,j} = 0$;\\
                 }
             }
             {
                 \For{$m = 1$ to $M$}{
                   Calculate $\Phi_{n}(n \in \mathcal{U}_{m,g}, \boldsymbol{\pi}_{-n})$ based on Eq. (\ref{eq13});\\
                 }
                 Find $m_{n}^{\ast} = \arg \min_{m \in \mathcal{M}} \Phi_{n}(n \in \mathcal{U}_{m,g}, \boldsymbol{\pi}_{-n})$;\\
                 Update the offloading strategy of device $n$ to $m_{n}^{\ast}$;\\
             }
           }
         }

         \For{$t = 1$ to $T$}{
           \For{$n = 1$ to $N$}{
                 \For{$g = 1$ to $G$}{
                    Calculate $\Phi_{n}(n \in \mathcal{U}_{m_{n}^{\ast},g}, \boldsymbol{\pi}_{-n})$ based on Eq. (\ref{eq13});\\
                 }
                 Find $g_{n}^{\ast} = \arg \min_{g \in \mathcal{G}} \Phi_{n}(n \in \mathcal{U}_{m_{n}^{\ast},g}, \boldsymbol{\pi}_{-n})$;\\
                 Update the grouping strategy of device $n$ to $g_{n}^{\ast}$;\\
             }
           }
    }
    \textbf{return} $\boldsymbol{\pi}^{\ast}$ \rm{and} $T_{n}, \forall n \in \mathcal{N}$.
\end{algorithm}

\par In the game $\Gamma$, any unilateral change in offloading and grouping strategies $(m,g)$ may lead to a decrease in the potential function for certain players, thus reducing their interference. However, it could also cause an increase in the potential function value of other players, thereby increasing their interference. According to \textbf{Definition \ref{def4}}, when the potential game $\Gamma$ reaches a NE state, the system can be considered Pareto optimal. Given that potential games can ensure the achievement of NE, we propose the potential game-based optimal joint offloading and grouping algorithm, as detailed in \textbf{Algorithm \ref{algo1}}, which aims to realize the optimal task offloading and user grouping strategies. \textbf{Algorithm \ref{algo1}} utilizes the properties of potential games to update players' offloading and grouping decisions iteratively, ultimately achieving a NE in the system. Thus, based on ensuring the interests of each player, the optimal offloading and grouping strategies of the entire system can be achieved. At the beginning of \textbf{Algorithm \ref{algo1}}, it first clears each group and initializes the offloading and grouping strategies for each device (\textbf{Line 1-2}). Before the game begins, the system first evaluates whether the device's computing capacity meets the delay requirements $D_{n}$. If the delay limit of $D_{n}$ is met, the device chooses local computing, which helps reduce device interference caused by offloading a large number of tasks to the same BS-MEC. If the computation capacity of the device itself is insufficient to meet the delay requirements, it chooses to offload to BS-MEC and then enters the game process corresponding to \textbf{Line 11-22}, which aims to determine the optimal offloading and grouping strategies for each device.

\subsubsection{\textbf{Optimality and Computational Complexity Analysis}} The proposed \textbf{Algorithm \ref{algo1}} is based on the framework of dynamic potential games. In each iteration, given the strategies $\boldsymbol{\pi}_{-n}$ of the other device $n \in \mathcal{N}$, autonomously selects its offloading and grouping strategy to minimize its interference, which is equivalently represented by the potential function. Since the formulated game is an EPG, any unilateral deviation by player $n$ yields an identical change in both its interference function and the global potential function. Formally, for any unilateral deviation we have (\ref{eq12}). Consequently, whenever a player strictly improves its interference function, the potential function also strictly decreases, as follows:

\vspace{-1em}

\begin{equation}\label{eq21}
    \begin{aligned}
        & \quad \nu_{n}\left(n \in \mathcal{U}_{m^{\prime},g^{\prime}}, \boldsymbol{\pi}^{\ast}_{-n} \right) < \nu_{n}\left(n \in \mathcal{U}_{m,g}, \boldsymbol{\pi}^{\ast}_{-n} \right)\\
        & \Leftrightarrow \Phi_{n}\left(n \in \mathcal{U}_{m^{\prime},g^{\prime}}, \boldsymbol{\pi}^{\ast}_{-n} \right) < \Phi_{n}\left(n \in \mathcal{U}_{m,g}, \boldsymbol{\pi}^{\ast}_{-n} \right).  \\
    \end{aligned}
\end{equation}

\par Since each player’s strategy set is finite, the global strategy space $\Pi$ is also finite. Therefore, the potential function can only take values from a finite set $\left\{\Phi(\boldsymbol{\pi}): \boldsymbol{\pi} \in \Pi\right\}$. As a result, an infinite strictly decreasing sequence of potential values cannot exist, which is known as the FIP described in \textbf{Definition \ref{def3}}, guaranteeing that the iterative process terminates in a finite number of steps. Once no player can further improve its interference function through unilateral deviation, the system reaches a NE. By \textbf{Definition \ref{def4}}, such an NE also satisfies the criterion of Pareto optimality, since no further Pareto improvements can be achieved at that point.

\par The main computational cost of \textbf{Algorithm \ref{algo1}} lies in \textbf{Steps 11–21}, where each device $n$ evaluates and updates its offloading and grouping strategies. The complexity of this evaluation is $\mathcal{O}\left(M+G\right)$. With $N$ devices in the system, and assuming that on average $C$ iterations are required for \textbf{Algorithm \ref{algo1}} to converge, then the overall computational complexity can be given by $\mathcal{O}\left(NC(M+G)\right)$. In the worst-case scenario, the maximum number of iterations $C$ is upper-bounded by the size of the joint strategy space, i.e., $C \leq |\Pi| - 1 = \Pi_{n=1}^{N} |S_{n}| - 1$, where $S_{n}$ denotes the strategy set of player $n$. However, in practice, the algorithm typically converges with $C$ much smaller than this bound.

\subsection{Convex Optimization-Based Power Allocation Scheme}

\vspace{-0.2em}

\par Given the task offloading and user grouping strategies $\boldsymbol{\pi}^{\ast}$ obtained by \textbf{Algorithm \ref{algo1}}, $\mathcal{P}1$ can be degenerated into the power allocation subproblem, as described in $\mathcal{P}2$.

\vspace{-1em}

\begin{subequations}\label{eq22}
\begin{align}
\mathcal{P}2:           & \min_{\{p_{n}\}_{\forall n \in \mathcal{N}}} \sum\limits_{n = 1}^{N} T_{n}, \\
\emph{s.t.} \quad \quad & 0 \leq p_{n} \leq p_{max}, \quad \forall n \in \mathcal{N}, \\
                        & 0 \leq T_{n} \leq D_{n}, \quad \forall n \in \mathcal{N},   \\
                        & E_{n} \leq E_{\mathrm{th}}, \quad \forall n \in \mathcal{N}.
\end{align}
\end{subequations}

\par $\mathcal{P}2$ is a non-convex function with respect to the power allocation scheme $\boldsymbol{p} = \{p_{1}, \cdots, p_{N}\}$, as it exists not only in the numerator but also in the denominator during computation. In order to overcome the solving challenges caused by non convexity, we transform the original minimization total delay into maximization of transmission rate, which can find the suboptimal solutions of $\mathcal{P}2$. This method utilizes the inverse relationship between maximizing transmission rate and minimizing delay, that is, indirectly reducing the total delay of computation and transmission by increasing transmission rate. After transforming the optimization objective, $\mathcal{P}2$ can be rewritten as follows:

\vspace{-1em}

\begin{subequations}\label{eq23}
\begin{align}
\mathcal{P}3:           & \max_{\{p_{n}\}_{\forall n \in \mathcal{N}}} \sum\limits_{n = 1}^{N} R_{n}^{mg}, \\
\emph{s.t.} \quad \quad & (\ref{eq22}b)-(\ref{eq22}d),
\end{align}
\end{subequations}
where the expression of $R_{n}^{mg}$ can be formulated as follows:

\vspace{-0.5em}

\begin{equation}\label{eq24}
   \begin{aligned}
     R_{n}^{mg} & = B \log_{2}\bigg(\sigma^{2} + \sum\limits_{\substack{i \in \mathcal{U}_{m,g} \\ |h^{mg}_{i}| \leq |h^{mg}_{n}|}} p_{i} |h^{mg}_{i}|^{2}\bigg) \\
     & - B \log_{2}\bigg(\sigma^{2} + \sum\limits_{\substack{i \in \mathcal{U}_{m,g} \\ |h^{mg}_{i}| < |h^{mg}_{n}|}} p_{i} |h^{mg}_{i}|^{2}\bigg).
    \end{aligned}
\end{equation}

\par $\mathcal{P}3$ can be effectively tackled using Majorization-Minimization (MM) algorithm \cite{10430407,9333915}, which can simplify the solution of non-convex problems in various ways, such as the Taylor expansion method, convex inequality method, etc. In this paper, we leverage the Taylor expansion method to optimize $\mathcal{P}3$, which can effectively transform non-convex functions into convex functions. The Taylor expansion method transforms the original non-convex problem into a convex problem by expanding it using the Taylor series. For example, given a first-order differentiable function $f$, its Taylor expansion at point $x_{0}$ can be formulated as

\begin{equation}\label{eq25}
     f(x) = f(x_{0}) + \nabla f(x_{0})\left(x - x_{0}\right).
\end{equation}

\par Based on the principle of Taylor expansion, the expression of $R^{mg}_{n}$ can be expanded as follows:

\begin{subequations}\label{eq26}
     \begin{align}
         \widetilde{R}_{n}^{mg} & = B \log_{2}\bigg(\sigma^{2} + \sum\limits_{\substack{i \in \mathcal{U}_{m,g} \\ |h^{mg}_{i}| \leq |h^{mg}_{n}|}} p_{i}^{(t)} |h^{mg}_{i}|^{2}\bigg) \nonumber \\
                                 & - B \log_{2}\bigg(\sigma^{2} + \sum\limits_{\substack{i \in \mathcal{U}_{m,g} \\ |h^{mg}_{i}| < |h^{mg}_{n}|}} p_{i}^{(t-1)} |h^{mg}_{i}|^{2}\bigg) \nonumber \\
                                 & - \frac{B}{\ln 2} \frac{\left(\sum\limits_{\substack{i \in \mathcal{U}_{m,g} \\ |h^{mg}_{i}| < |h^{mg}_{n}|}} \left(p_{i}^{(t)} - p_{i}^{(t-1)}\right) |h^{mg}_{i}|^{2}\right)}{\left(\sigma^{2} + \sum\limits_{\substack{i \in \mathcal{U}_{m,g} \\ |h^{mg}_{i}| < |h^{mg}_{n}|}} p_{i}^{(t-1)}|h^{mg}_{i}|^{2}\right)}.
     \end{align}
\end{subequations}

\par After Taylor expansion, $\widetilde{R}^{mg}_{n}$ is already a convex function, thus $\mathcal{P}3$ can be transformed into a convex problem, as described in $\mathcal{P}4$, as follows:

\begin{subequations}\label{eq27}
\begin{align}
\mathcal{P}4:           & \max_{\{p_{n}\}_{\forall n \in \mathcal{N}}} \sum\limits_{n = 1}^{N} \widetilde{R}^{mg}_{n}, \\
\emph{s.t.} \quad \quad & (\ref{eq22}b)-(\ref{eq22}d).
\end{align}
\end{subequations}

\vspace{-1em}

\begin{algorithm}[htbp]
 \setstretch{1}
        \caption{Majorization-Minimization (MM) Based Power Allocation Algorithm.}
        \label{algo2}
        \KwIn{Device set $\mathcal{N}$; BS-MEC set $\mathcal{M}$; Subchannel set $\mathcal{G}$; Channel coefficient $h^{m g}_{n}$; Delay constraints $D_{n}$; Task offloading and grouping strategies $\boldsymbol{\pi}$; Convergence threshold $\varepsilon$; Maximum iteration numbers $T_{max}$;}
        \KwOut{Optimal power allocation $\boldsymbol{p}^{\ast} = \{p_{1}^{\ast}, \cdots, p_{N}^{\ast}\}$.}
        Under the premise of ensuring constraints (\ref{eq22}b), (\ref{eq22}c) and (\ref{eq22}d), initialize the power allocation scheme of devices;\\
        Initialize the number of iteration $t = 1$;\\
        \While{$\left|\widetilde{R}^{mg^{(t)}}_{n} - \widetilde{R}^{mg^{(t-1)}}_{n}\right| > \varepsilon$   \text{\rm{and}} $t < T_{max}$}{
           Using CVX to solve problem $\mathcal{P}4$ and obtaining the power allocation scheme $\boldsymbol{p}^{(t)}$;\\
           t = t + 1;\\
        }
\end{algorithm}

\vspace{-1em}

\par After the transformation of Taylor expansion, the original power allocation problem $\mathcal{P}2$ is effectively transformed into a convex optimization problem $\mathcal{P}4$, which can be solved using the existing convex optimization toolkit, like CVX. Moreover, considering that the objective function $\widetilde{R}^{mg}_{n}$ is obtained through an iterative process, this paper needs to use an iterative algorithm to gradually approach the optimal power allocation scheme, in order to ensure that the algorithm can converge to the optimal solution while meeting certain accuracy requirements. Specifically, the detailed steps of the power allocation algorithm are detailed in \textbf{Algorithm \ref{algo2}}.

\vspace{-1em}

\subsection{AO-Based Joint Task Offloading, User Grouping, and Power Allocation Scheme}

\par In this subsection, we propose the alternating optimization (AO)-based joint task offloading, user grouping, and power allocation algorithm to effectively address $\mathcal{P}1$. Firstly, we fix the power allocation scheme $\boldsymbol{p} = \{p_{n}\}_{n \in \mathcal{N}}$, and then invoke \textbf{Algorithm \ref{algo1}} to derive the optimal task offloading and user grouping strategies, i.e., $\boldsymbol{a} = \{a_{m,n}\}_{m \in \mathcal{M}, n \in \mathcal{N}}$ and $\boldsymbol{b} = \{b_{g,n}\}_{g \in \mathcal{G}, n \in \mathcal{N}}$, respectively. Secondly, we use the derived $\left [\boldsymbol{a}, \boldsymbol{b}\right ]$ as the inputs for Algorithm \ref{algo2}, and then execute it to update the power allocation scheme $\boldsymbol{p}$. Loop iteratively optimize the above process until $\left[\boldsymbol{a}, \boldsymbol{b}, \boldsymbol{p}\right]$ converges to stable solutions or the termination conditions are satisfied. The detailed process mentioned above can be described by \textbf{Algorithm \ref{algo3}}.

\begin{algorithm}[h]
\small
 \setstretch{1.0}
        \caption{AO-Based Joint Task Offloading, User Grouping, and Power Allocation Scheme.}
        \label{algo3}
        \KwIn{IoT device set $\mathcal{N}$, BS set $\mathcal{M}$, subchannel set $\mathcal{G}$, channel coefficient $h^{mg}_{n}$, delay constraints $D_{n}$, convergence threshold $\varepsilon$, maximum iteration numbers $T_{max}$;}
        \KwOut{Optimal task offloading and grouping strategies $\boldsymbol{\pi}^{\ast} = [\boldsymbol{a}^{\ast}, \boldsymbol{b}^{\ast}]$, and optimal power allocation scheme $\boldsymbol{p}^{\ast}$}
        \For{$n = 1$ to $N$}{
          Each device $n$ is connected to the nearest BS;\\
          Each device $n$ randomly selects subchannels;\\
          Initializes power allocation $p_{n}$;\\
        }
        \Repeat{All devices do not change their offloading and grouping strategies or reach the maximum number of iterations $T_{max}$}{
           Execute Algorithm \ref{algo1} to derive the offloading and grouping strategies $\boldsymbol{\pi}^{\ast} = [\boldsymbol{a}^{\ast}, \boldsymbol{b}^{\ast}]$;\\
           Execute Algorithm \ref{algo2} to obtain the power allocation scheme $\boldsymbol{p}^{\ast}$;\\
        }
    \textbf{return} $\left[\boldsymbol{a}^{\ast}, \boldsymbol{b}^{\ast}, \boldsymbol{p}^{\ast}\right]$.
\end{algorithm}

\vspace{-1em}

\section{Performance Evaluation}

\par In this section, we present extensive simulation experiments and discussions to validate the effectiveness of our proposed joint task offloading, user grouping, and power allocation scheme. We compare our algorithms against state-of-the-art methods and classic heuristic approaches, and provide detailed comparative analyses. Moreover, comprehensive ablation studies are conducted to underscore the key contributions of this work and confirm the overall superiority of our schemes.

\vspace{-0.3em}

\begin{table}[h]
 \renewcommand{\arraystretch}{1.2}
 \caption{\small Experimental Parameter Settings}
 \vspace{-0.5em}
 \label{tab1}
 \centering
 \resizebox{0.99\columnwidth}{!}
 {
 \begin{tabular}{|l|c|}
  \hline
  \bfseries Parameter                                               & \bfseries Value\\
  \hline
  Number of BS-MEC $M$                                              & $4$     \\
  \hline
  Number of IoT device $N$                                          & $80$    \\
  \hline
  Number of subchannels $G$                                         & $5$     \\
  \hline
  Bandwidth of each subchannel $B$                                  & $410$ kHz \\
  \hline
  Large-scale fading                                                & $128.1 + 37.6 \log(d_{m,n} [\rm{km}])$\\
  \hline
  Noise power spectral density $N_{0}$                              & $-174 \rm{dBm/Hz}$   \\
  \hline
  Size of tasks $W_{n}$                                             & $[5, 15]$ Mbits\\
  \hline
  Delay Constraints $D_{n}$                                         & $[0.2, 1.2]$ s\\
  \hline
  Computing capacity of IoT devices $f_{n}^{\rm{local}}$            & $[3, 8]$ Mbits/s\\
  \hline
  Computing capacity of MEC servers $f_{m}$                         & $[7, 10]$ Gbits/s\\
  \hline
  Maximum transmit power of devices $p_{max}$                       & $27.8$ dBm \\
  \hline
  Energy constraint $E_{\rm{th}}$                                   & $[0.152,0.910]$ J\\
  \hline
 \end{tabular}
 }
\end{table}

\vspace{-1.3em}

\subsection{Experimental Parameter Settings}

\par In our considered simulation scenario, IoT devices are randomly distributed within a rectangular area of $1000\ \rm{m} \times 1000\ \rm{m}$, with four BSs located at coordinate points $(250\ \rm{m} , 750\ \rm{m})$, $(250\ \rm{m} , 250\ \rm{m})$, $(750\ \rm{m} , 750\ \rm{m})$, and $(750\ \rm{m} , 250\ \rm{m})$, respectively. We consider that the small-scale fading follows a Rayleigh distribution with the mean of $0$ and variance of $1$, and follows an independent and identically distributed (i.i.d.) Gaussian distribution. The noise power can be calculated according to $\sigma^{2} = B N_{0}$. The specific settings of other simulation parameters are shown in Table \ref{tab1}.

\vspace{-0.3em}

\subsection{Baseline Schemes for Comparison}

\par To comprehensively evaluate the effectiveness of the proposed joint task offloading, user grouping, and power allocation scheme (denoted \textbf{Proposed}), we compare it against several state-of-the-art methods and classic heuristic approaches. In particular, for the user grouping in NOMA-assisted systems, we select two of the most popular grouping strategies, as follows:

\begin{itemize}
    \item Gale-Shapley Grouping (\textbf{Gale-Shapley}) \cite{7972929,9887894,9085263}: In this strategy, devices and groups are regarded as two sets with their own preference lists, with each device tending to access the groups with better channel conditions, and each group also tending to select the devices with higher channel gains.

    \item Max-Min based Grouping (\textbf{Max-Min}) \cite{10093902,8365852}: Each IoT device is grouped into the group with the greatest difference in channel conditions, so that the NOMA can decode the messages of IoT devices by leveraging the differences of channel conditions among IoT devices.
\end{itemize}

\par For the comparison of task offloading, we select two of the popular task offloading strategies, as follows:

\begin{itemize}
    \item Nearby BS-based Offloading (\textbf{Nearby}) \cite{10309191}:  Each IoT device will prioritize offloading its tasks to the nearest BS and select the sub-channel with the least interference for access.

    \item Computing Capacity-based Offloading (\textbf{Computing}): In this strategy, each IoT device offloads its task to the BS-MEC with the strongest computing capacity that currently allows access devices (i.e., those that have not yet reached the maximum device access limit).
\end{itemize}

\vspace{-0.5em}

\begin{figure}[htbp]
\centering
\includegraphics[scale=0.42]{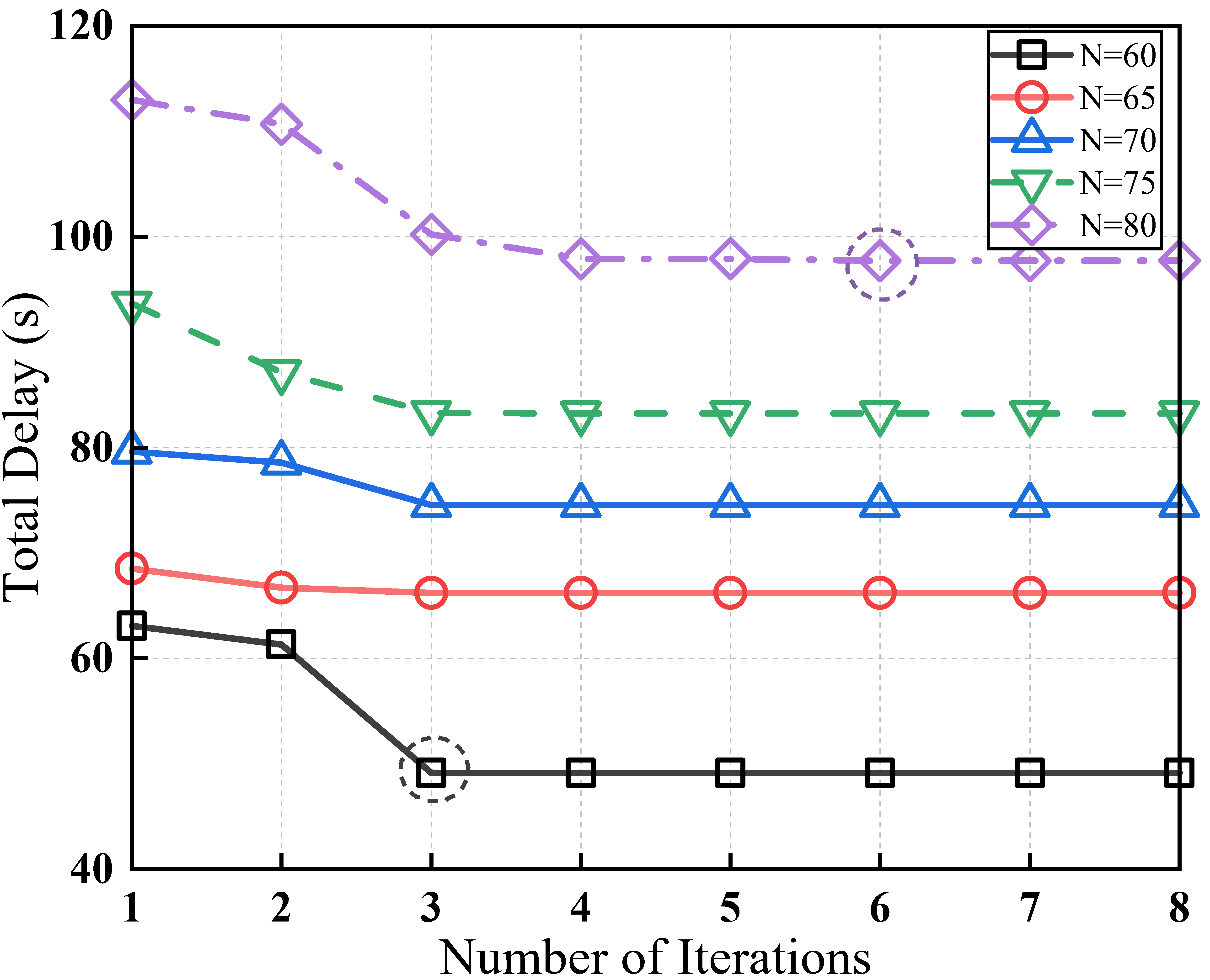}
\caption{\small The convergence performance of the proposed joint task offloading, user grouping, and power allocation scheme under different numbers of IoT devices.}
\label{fig_convergence}
\end{figure}

\subsection{Analysis of Experimental Results}

\subsubsection{Convergence Verification}

\par As shown in Fig. \ref{fig_convergence}, we demonstrate the convergence performance of the proposed joint task offloading, user grouping, and power allocation scheme under different numbers of IoT devices. As the number of IoT devices $N$ increases, the number of iterations required for the algorithm to converge also increases. In particular, when $N=60$, the total delay of the system can converge to a stable value after $3$ iterations. However, when $N=80$, the number of iterations required for the algorithm to converge increases to $6$. This is mainly because, as the number of devices increases, the decision space available for each device to select from when making offloading and grouping strategies expands, which may require more iterations to achieve convergence. However, the average number of iterations of the proposed algorithm always remains below $6$ times when no more than $80$ IoT devices are connected, which fully demonstrates that our algorithm can quickly converge to the NE state in a short number of iterations.

\begin{figure}[htbp]
\centering
\includegraphics[scale=0.68]{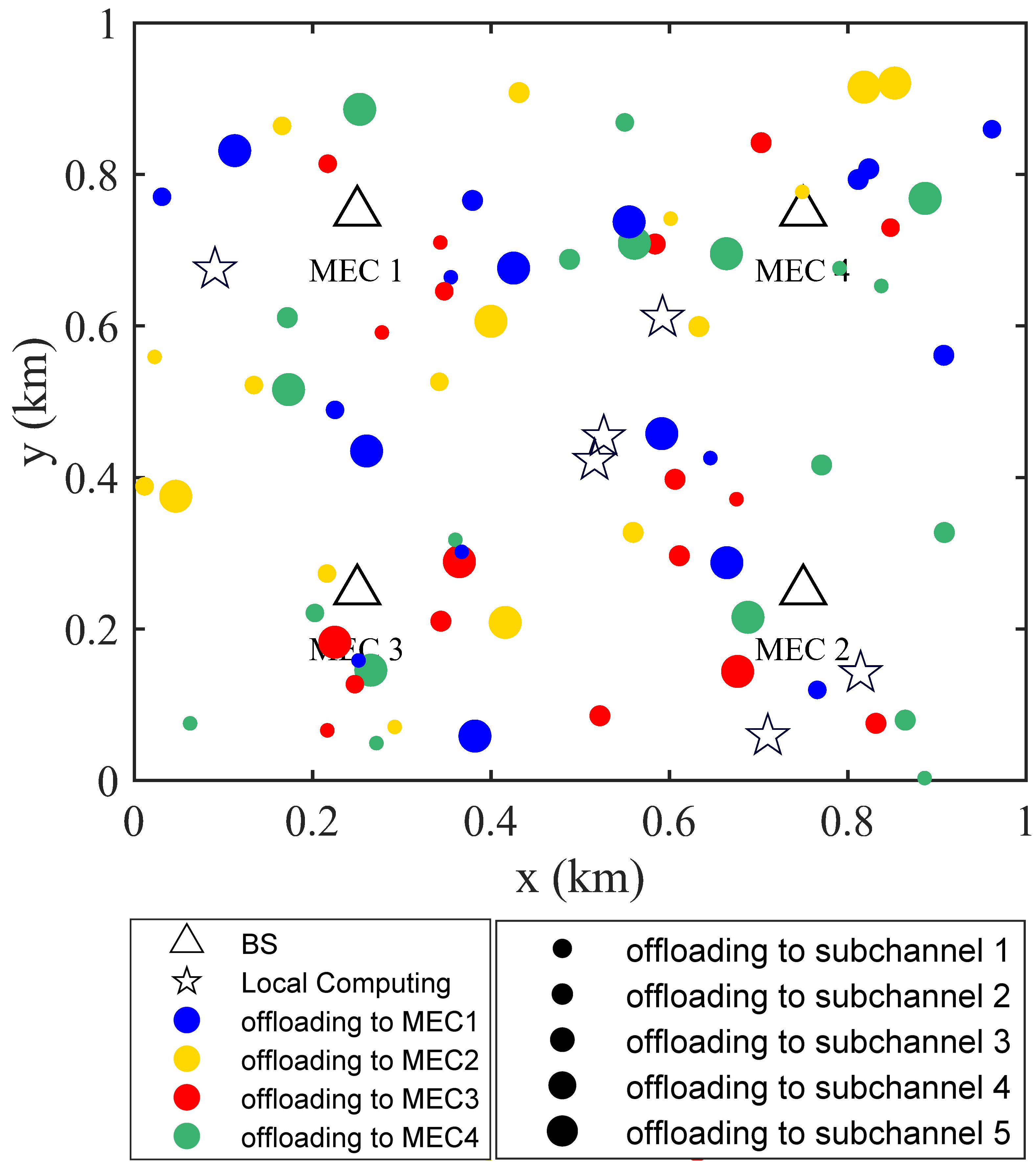}
\caption{\small The distribution of IoT device locations after convergence of the proposed joint task offloading, user grouping, and power allocation scheme.}
\label{fig_Convergence_Locations}
\end{figure}

\par As shown in Fig. \ref{fig_Convergence_Locations}, we present the visualization results of the device distribution and the task offloading and user grouping strategies of the algorithm after convergence. The observation results show that in the final offloading decision of each device, the device does not always select to offload its tasks to the nearest BS-MEC. This phenomenon explains that, after considering multiple factors, such as channel conditions, computing requirements, and BS-MEC computing capabilities, an IoT device may make offloading decisions that favor a non-nearest BS in order to optimize overall system performance, rather than solely determining its own grouping strategy based on the quality of channel conditions.

\subsection{Performance Comparison}

\par As illustrated in Fig.~\ref{fig_total_delay_average_delay_vs_IoT_device_number}, we present the total delay and average delay performance of different algorithms under various IoT device access scenarios. Under a fixed number of BS-MECs $M = 4$ and subchannels $G = 5$, it can be observed that both the total delay and average delay per device increase with the value of $N$. This is mainly because that as the number of IoT devices increases, the number of devices offloaded to the same BS-MEC also increases accordingly, leading to intensified intra-group interference among devices and consequently reducing the uplink transmission throughput. On the other hand, since the computing capacity of the BS-MEC server is also limited, the increase in the number of IoT devices leads to more intense competition for computing resources among devices, which in turn results in higher computational delays.

\par In addition, Fig.~\ref{fig_total_delay_average_delay_vs_IoT_device_number} (a) and (b) reveal that the proposed algorithm outperforms all other comparison algorithms, including Gale-Shapley, Max-Min, Nearby, and Computing, in terms of both total delay and average delay performance. In particular, our proposed schemes in this paper reduce the delay by at least 19.3\% compared to these comparison algorithms, and up to 41.5\%. These experimental results demonstrate that the proposed algorithm has significant advantages in achieving low latency when dealing with the IoT device access with DSCI tasks.

\begin{figure*}[t]
    \centering
    \begin{minipage}[b]{0.47\textwidth}
        \centering
        \begin{subfigure}{\linewidth}
            \centering
            \includegraphics[width=\linewidth]{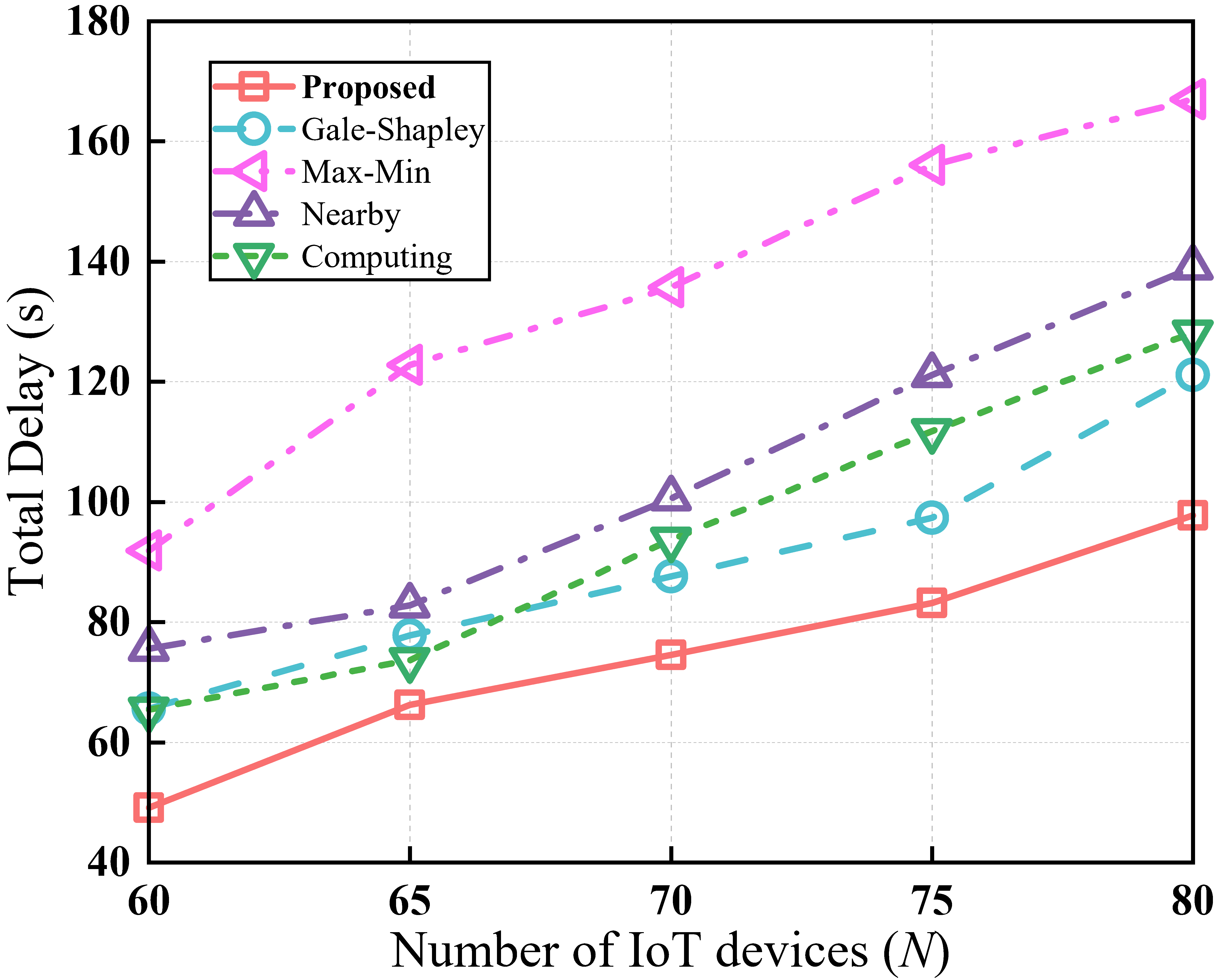}
            \caption{Total delay versus number of IoT devices.}
        \end{subfigure}
    \end{minipage}
    \hspace{0.005\textwidth}
    \begin{minipage}[b]{0.48\textwidth}
        \centering
        \begin{subfigure}{\linewidth}
            \centering
            \includegraphics[width=\linewidth]{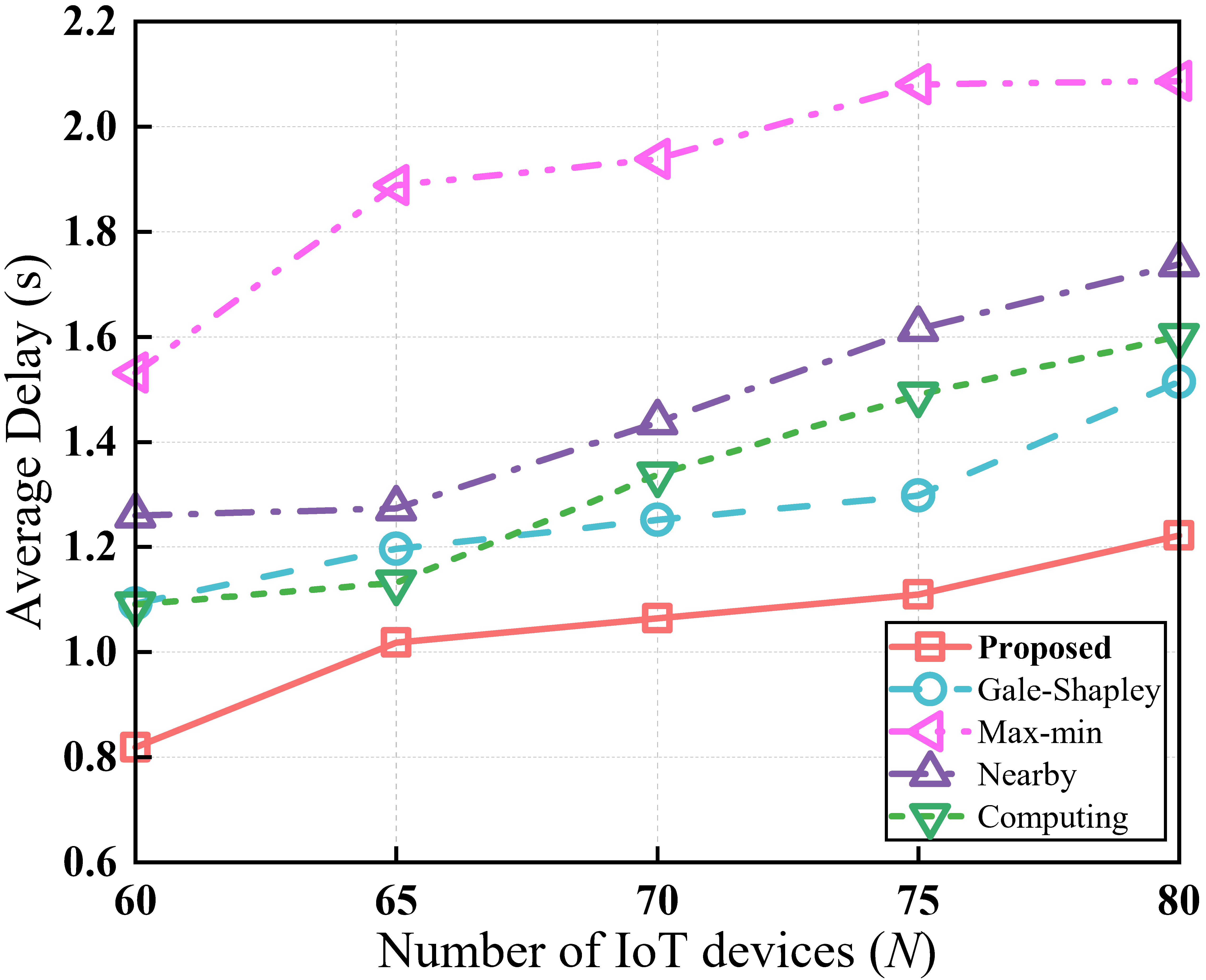}
            \vspace{-1.5em}
            \caption{Average delay versus number of IoT devices.}
        \end{subfigure}
    \end{minipage}
    \caption{\small Total and average delay performance under different IoT device access scenarios.}
    \label{fig_total_delay_average_delay_vs_IoT_device_number}
\end{figure*}

\par As shown in Fig.~\ref{fig_total_power_average_power_vs_IoT_device_number}, we further investigate the total power consumption and average power consumption corresponding to different algorithms under various IoT device access scenarios. Under the set conditions of $M = 4$ and $G = 5$, it can be observed that the system's total power consumption increases with the number of IoT devices. However, our proposed algorithm consistently maintains the lowest power consumption overhead, demonstrating its superior energy efficiency. Specifically, when $N = 80$, the proposed algorithm saves approximately 6.7\% of power consumption compared to the Gale-Shapley, which has the second-best delay performance in Fig.~\ref{fig_total_power_average_power_vs_IoT_device_number}. Furthermore, we study the average power consumption per IoT device under different IoT device access scenarios. It is evident that our proposed algorithm achieves significant improvements in average power consumption per IoT device compared to the Gale-Shapley, Max-Min, Nearby, and Computing. For example, in the scenario of IoT device access with DSCI tasks, the proposed algorithm achieves at least 14.7\% reduction in average power consumption per IoT device.

\begin{figure*}[t]
    \centering
    \begin{minipage}[b]{0.47\textwidth}
        \centering
        \begin{subfigure}{\linewidth}
            \centering
            \includegraphics[width=\linewidth]{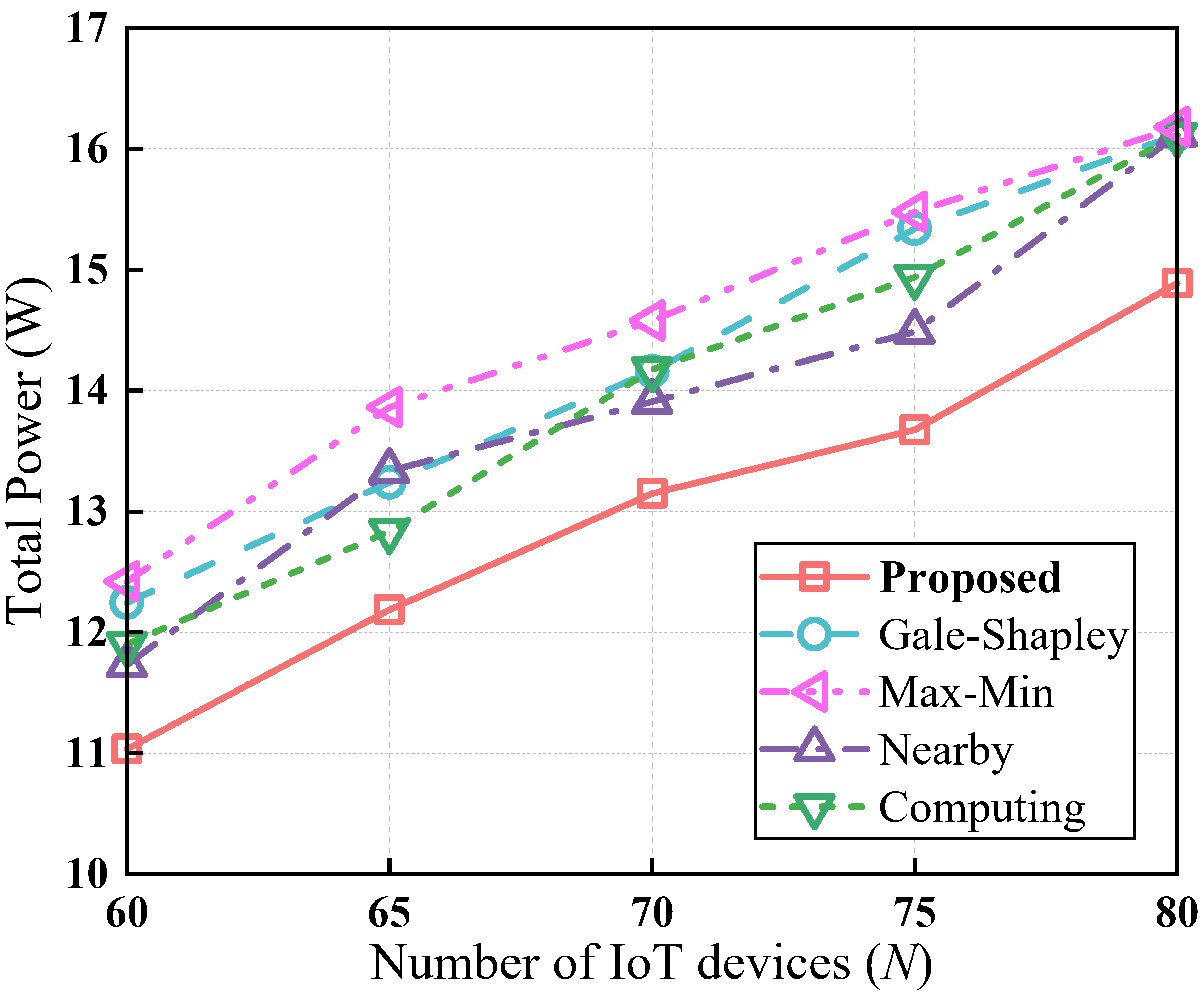}
            \caption{Total power versus number of IoT devices.}
        \end{subfigure}
    \end{minipage}
    \hspace{0.005\textwidth}
    \begin{minipage}[b]{0.49\textwidth}
        \centering
        \begin{subfigure}{\linewidth}
            \centering
            \includegraphics[width=\linewidth]{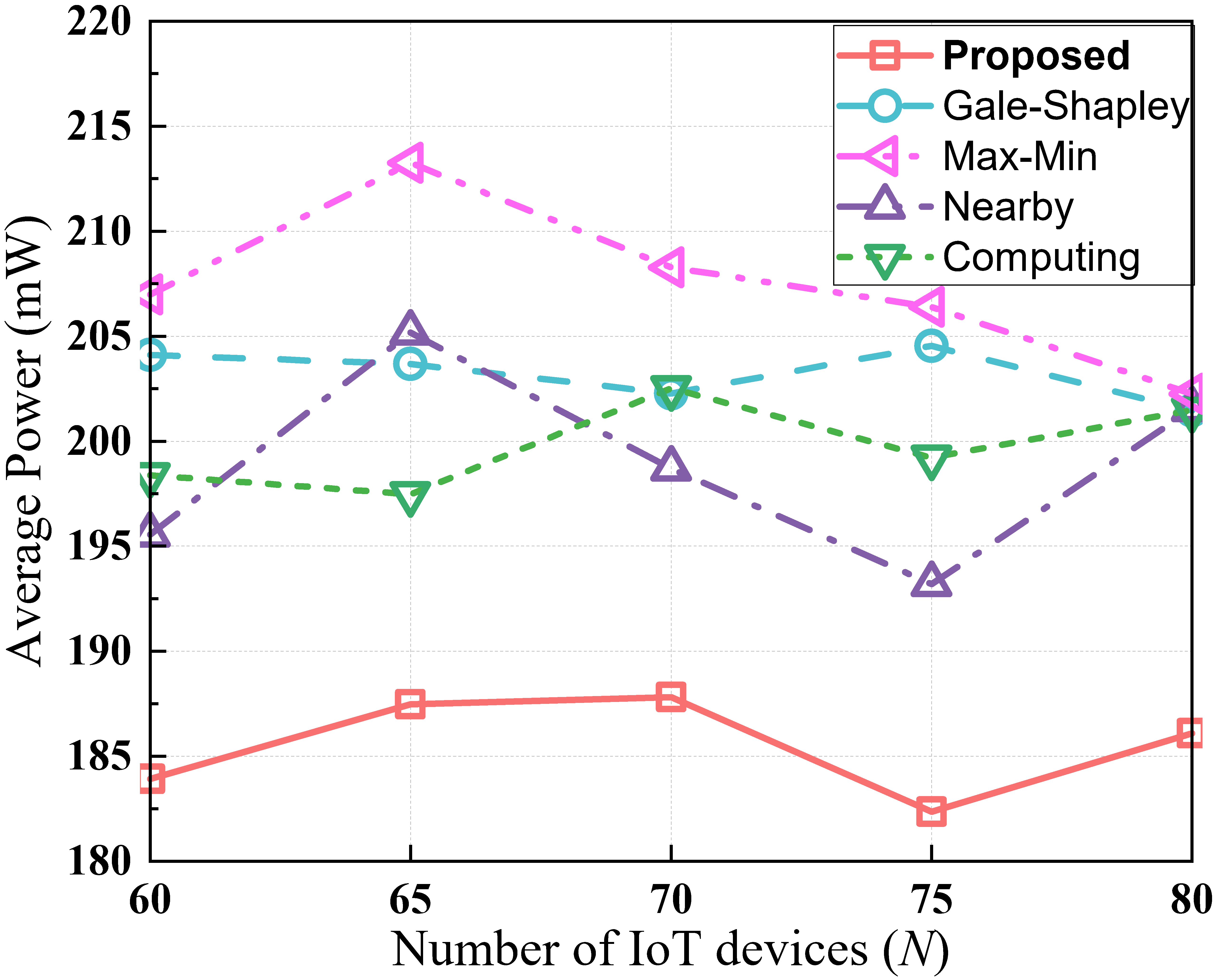}
            \vspace{-1.5em}
            \caption{Average power versus number of IoT devices.}
        \end{subfigure}
    \end{minipage}
    \caption{\small Total and average power consumption performance under different IoT device access scenarios.}
    \vspace{-1.5em}
    \label{fig_total_power_average_power_vs_IoT_device_number}
\end{figure*}

\par We provide a more detailed analysis of the performance gains in delay and power consumption achieved by our proposed algorithm, as shown in Fig.~\ref{fig_total_delay_average_delay_vs_IoT_device_number} and Fig.~\ref{fig_total_power_average_power_vs_IoT_device_number}. Unlike the comparison algorithms, the proposed algorithm allows for a flexible number of IoT devices in different groups, enabling more dynamic decision-making. In the grouping process, our algorithm takes into account not only the device's channel conditions but also intra-group interference, interference between different BSs, and the power allocation strategy, making the IoT device offloading and grouping decisions more superior than the selected schemes based solely on channel conditions or distance. Moreover, in our algorithm, IoT devices dynamically determine their own strategies, unlike other comparison algorithms that determine all devices' strategies at once. This allows our approach to fully consider the impact of other devices on each device's performance. In contrast, algorithms like Gale-Shapley, Nearby, and Max-Min primarily prioritize based on device channel conditions, neglecting the significant effects of intra-group interference and power distribution on throughput (and thus on delay). The offloading strategies of the selected Nearby may result in too many tasks being offloaded to the same BS-MEC, leading to an imbalanced computational load and a significant increase in intra-group interference, which impacts throughput and increases the system's computational burden. The Computing algorithm, on the other hand, mainly focuses on the computation requirements of device tasks, overlooking the impact of intra-group interference, power allocation, and other key factors. By comprehensively considering factors such as channel conditions, intra-group and inter-base station interference, and power allocation, our algorithm presents a more refined and efficient task offloading and grouping strategy, significantly enhancing system performance and ensuring lower delays and higher computational efficiency.

\begin{figure*}[t]
    \centering
    \begin{minipage}[b]{0.47\textwidth}
        \centering
        \begin{subfigure}{\linewidth}
            \centering
            \includegraphics[width=\linewidth]{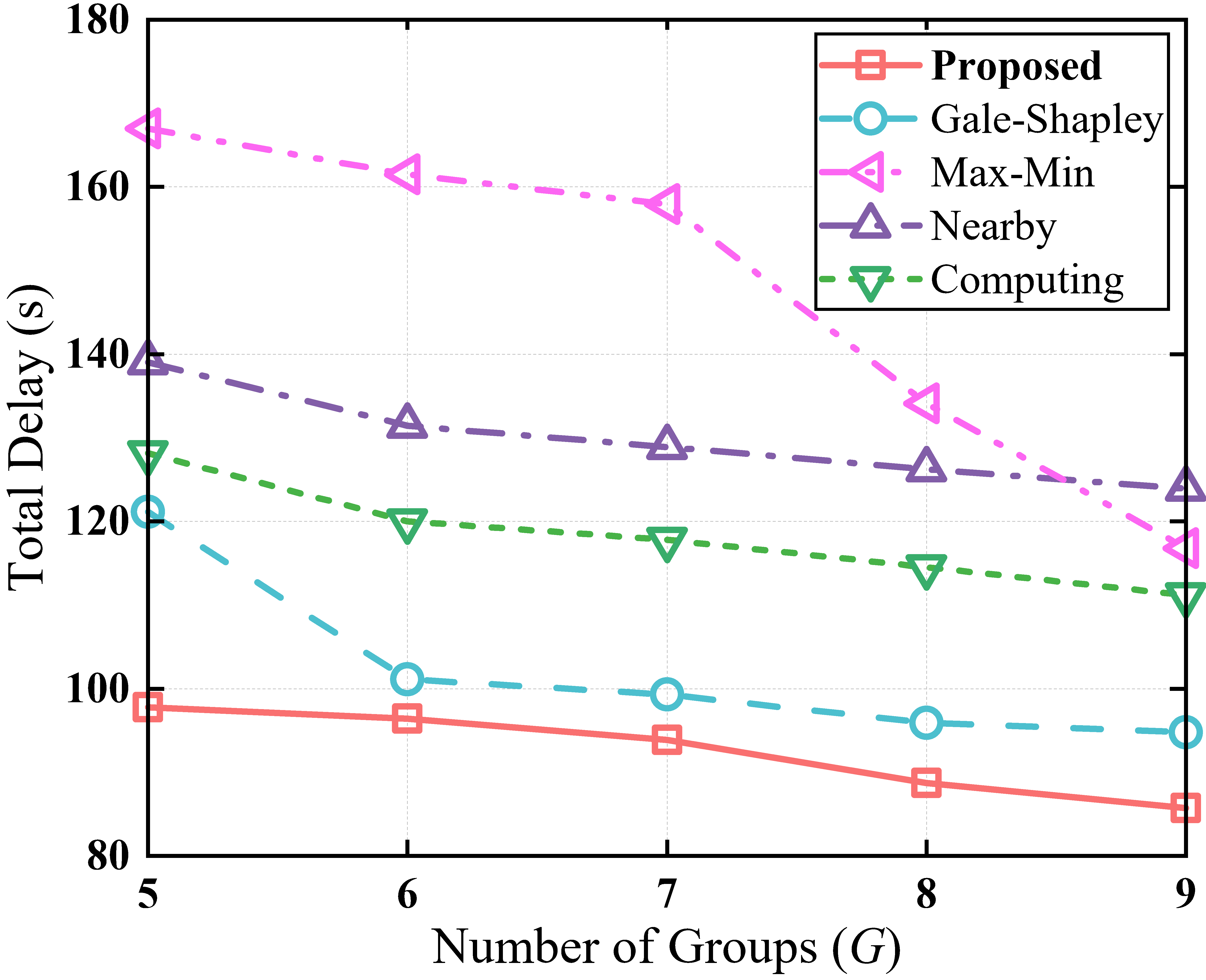}
            \caption{Total delay versus number of groups.}
        \end{subfigure}
    \end{minipage}
    \hspace{0.005\textwidth}
    \begin{minipage}[b]{0.49\textwidth}
        \centering
        \begin{subfigure}{\linewidth}
            \centering
            \includegraphics[width=\linewidth]{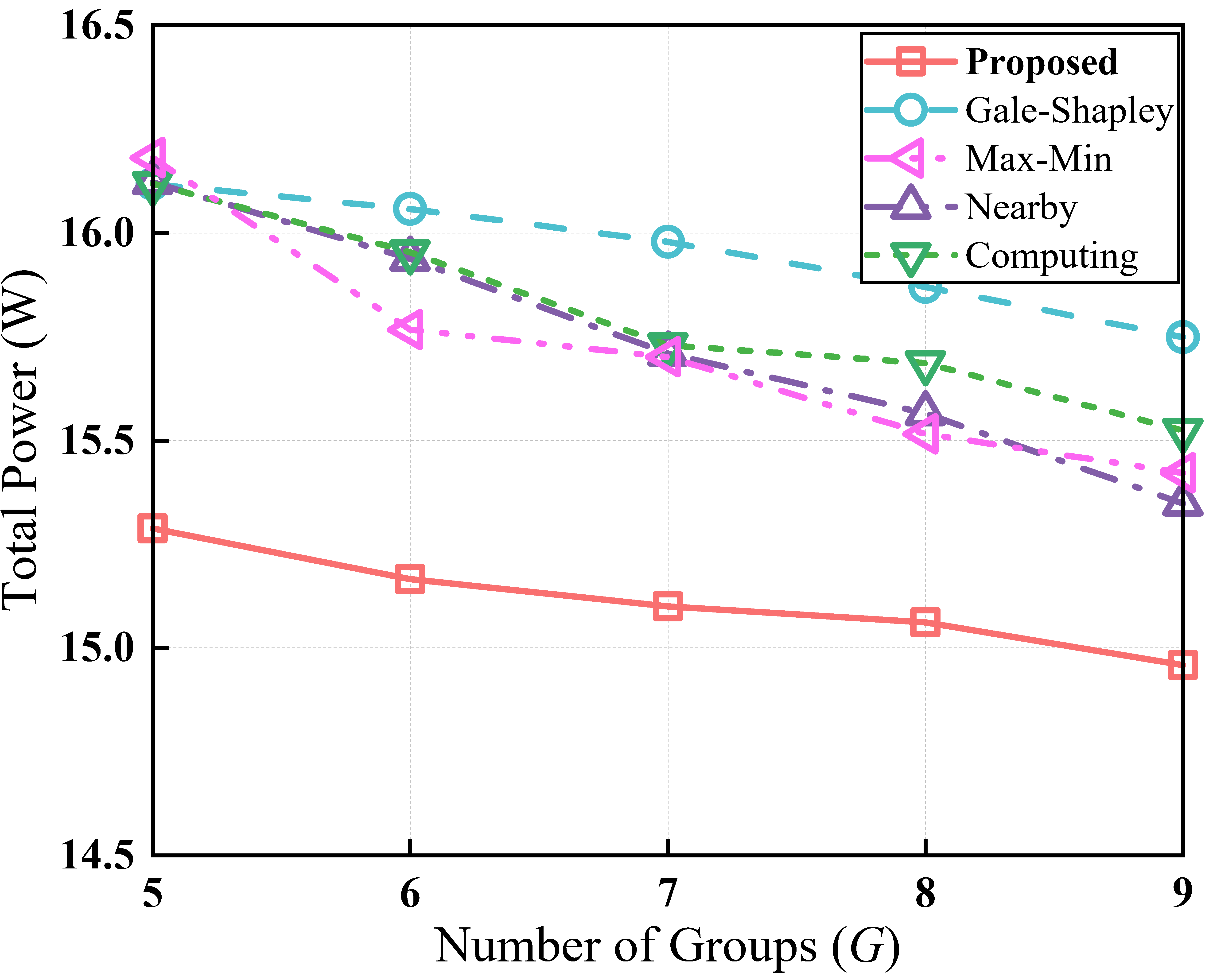}
            \vspace{-1.5em}
            \caption{Total power versus number of groups.}
        \end{subfigure}
    \end{minipage}
    \caption{\small Performance variations of system's total delay and total power consumption under different number of groups in IoT access scenarios with DSCI tasks. ($N = 80$).}
    \vspace{-1.5em}
    \label{fig_total_delay_total_power_vs_group_numbers}
\end{figure*}

\par As shown in Fig.~\ref {fig_total_delay_total_power_vs_group_numbers}, we studied the performance of the system's total delay and total power consumption under different numbers of groups $G$ (i.e., number of sub-channels) under large-scale access of IoT devices ($N = 80$). As shown in Fig.~\ref {fig_total_delay_total_power_vs_group_numbers} (a), with the increase of the number of groups, the delay performance corresponding to different algorithms has been correspondingly improved. However, our proposed algorithm achieves the best latency performance compared to Gale Shapley, Max Min, Nearby, and Computing. This is mainly because, as the number of sub-channels increases, IoT devices offloaded to the same BS-MEC will have more group options, resulting in a corresponding decrease in the number of IoT devices within each group. This change means that intra-group interference will be improved to a certain extent, which in turn will lead to an increase in throughput and ultimately improve the latency performance of the system.

\par As shown in Fig.~\ref{fig_total_delay_total_power_vs_group_numbers} (b), this paper also investigates the impact of changes in the number of groups on total power consumption. It can be observed that as the number of groups increases, the total power consumption of the system shows a decreasing trend. This is because the increase in the number of groups leads to a decrease in the number of devices in each group, thereby reducing intra-group interference. At this point, lower power can be allocated to meet the delay requirements of each device. This result indicates that the algorithm proposed in this paper can not only effectively reduce system latency compared to the comparative algorithm, but also reduce system power consumption.

\begin{figure}[htbp]
\centering
\includegraphics[scale=0.42]{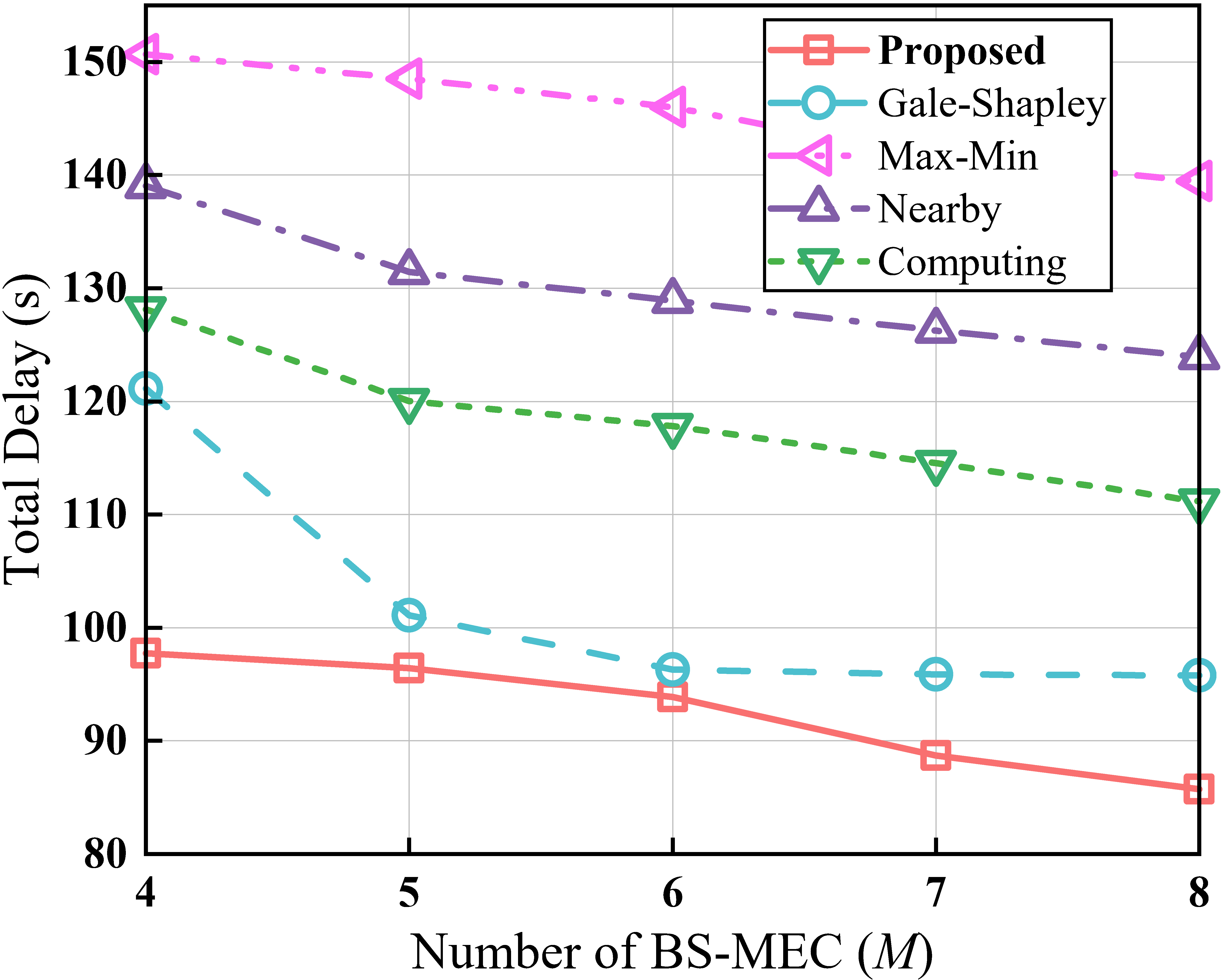}
\caption{\small Performance variations of system's total delay under different number of BS-MECs in IoT access scenarios with DSCI tasks ($N = 80$).} %
\label{fig_total_delay_versus_BS_MEC}
\end{figure}

\par As shown in Fig.~\ref{fig_total_delay_versus_BS_MEC}, we demonstrate the impact of varying the number of BS-MECs on the system's total delay performance in IoT access scenarios with DSCI tasks. It can be observed that as the number of BS-MECs increases, the total delay performance corresponding to different algorithms improves accordingly. Notably, the proposed algorithm consistently maintains the best delay performance compared to the benchmark algorithms, including Gale-Shapley, Max-Min, Nearby, and Computing, when the number of BS-MECs increases from $M = 4$ to $M = 8$. However, for the Gale-Shapley algorithm, its delay improvement becomes less significant when $M \geq 6$, primarily due to the inherent design of the Gale-Shapley algorithm itself. Initially, as the number of BS-MECs increases, the system benefits from more distributed computing resources and reduced load per BS, thereby improving the system's delay performance. However, given a fixed number of IoT devices, once the number of BS-MECs exceeds a certain threshold, further increasing the number of BSs does not significantly alleviate the load on the BS-MEC nodes, because the number of IoT devices $N$ and the number of sub-channels $G$ are fixed. The Gale-Shapley algorithm prioritizes allocation based on channel gains, which may lead to an uneven distribution of devices among BSs, especially when high-gain devices are concentrated in a few BSs with better channel conditions. This concentration results in continuous intra-group interference and causes some BSs to become overloaded, creating computational bottlenecks that offset the potential delay reduction benefits of adding more BSs.

\begin{figure}[htbp]
\centering
\includegraphics[scale=0.42]{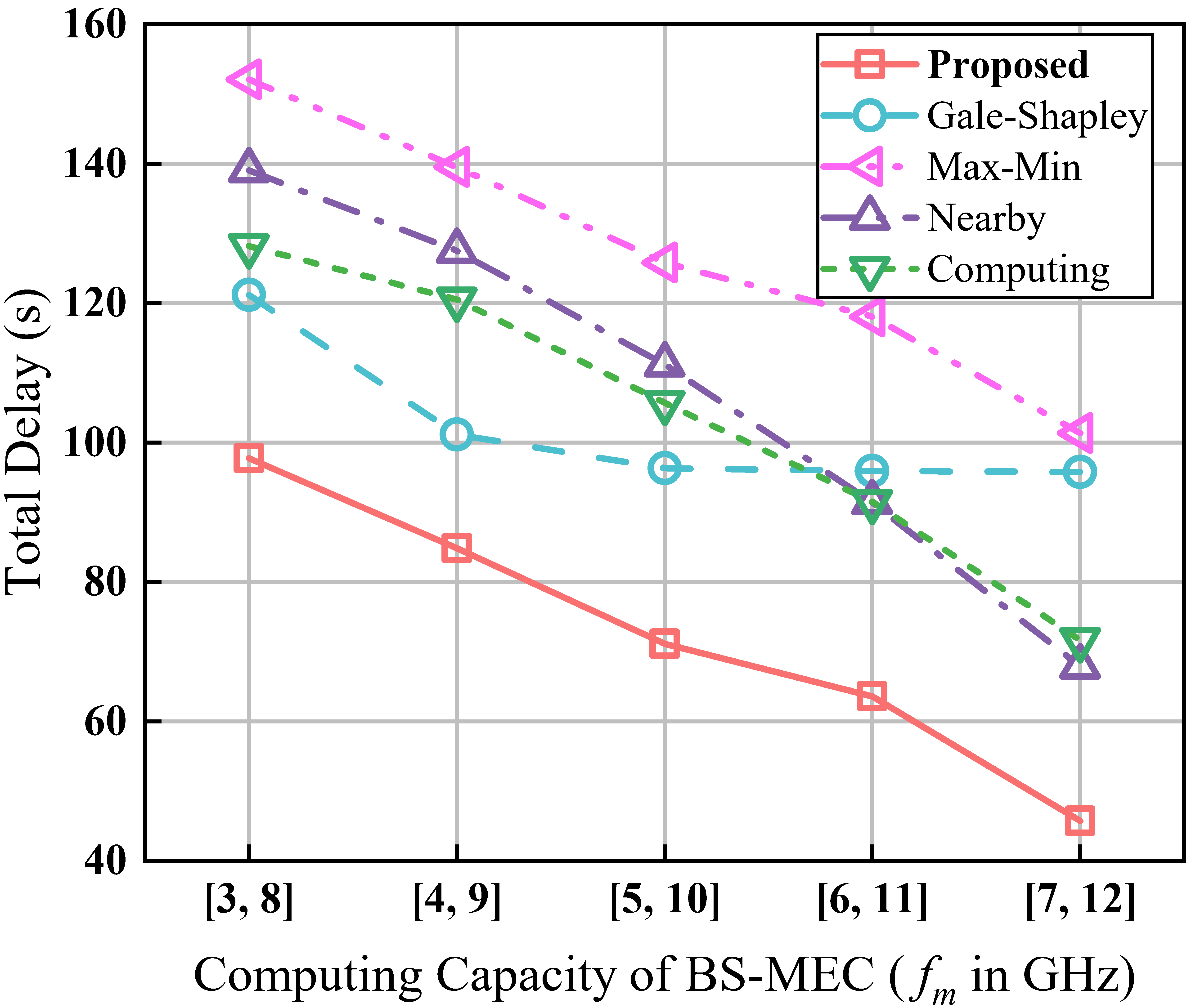}
\caption{\small The impact of BS-MEC's computing capacity on system's total delay performance in IoT access scenarios with DSCI tasks.}
\vspace{-1.5em}
\label{fig_total_delay_versus_BS_MEC_computing_capacity}
\end{figure}

\par As shown in Fig.~\ref{fig_total_delay_versus_BS_MEC_computing_capacity}, we investigated the impact of computing capacity of BS-MEC on the system's total delay performance in IoT access scenarios with DSCI tasks. It can be observed that as the computing capacity of each BS-MEC increases, the system's total delay corresponding to different algorithms improves accordingly. Notably, it can be noticed that when the computing capacity of BS-MECs reaches a certain range, the delay performance gain of the Gale-Shapley algorithm similarly reaches a bottleneck. This is mainly because the preference-based matching approach of the Gale-Shapley algorithm leads to an uneven distribution of devices, where high-gain devices tend to concentrate in groups associated with BS-MECs that have favorable channel conditions. Therefore, even with increased computing resources, these groups with those overloaded BS-MEC still experience severe intra-group interference and resource contention, which limits throughput and maintains high communication delay. This imbalance prevents a significant reduction in delay despite enhanced computing capabilities. In addition, it can be observed that under different ranges of BS-MEC computing capacities, the proposed algorithm consistently maintains superior performance gains compared to other benchmark algorithms, which fully demonstrates the effectiveness of our proposed algorithm.

\section{Conclusion and Future Outlook}
This paper proposed an uplink NOMA-assisted multi-BS-MEC network to fulfill the QoS requirements of massive IoT device connectivity with DSCI tasks. To effectively address the challenges of unbalanced subchannel access, inter-group interference, computational load disparities, and device heterogeneity in multi-BS MEC systems, this paper leverages game theory to propose an innovative joint task offloading and user grouping algorithm that dynamically makes optimal offloading-grouping decisions based on the channel states of individual IoT devices. Secondly, this paper introduces an MM-based power allocation algorithm that transforms the original power allocation problem into a tractable convex optimization paradigm, thereby deriving the optimal power allocation strategy. Moreover, an efficient iterative optimization algorithm is further designed to minimize the total system delay. Extensive simulation experiments have verified the effectiveness of our proposed optimization schemes, demonstrating performance improvements of up to 19.3\% and 14.7\% in terms of the system's total delay and power consumption, respectively.

\par Although the current work considers an uplink NOMA-assisted multi-MEC framework for large-scale IoT networks with DSCI tasks, each device is restricted to associating with a single BS, which severely limits mobility support. In future work, we will incorporate service migration mechanisms, enabling devices to dynamically offload tasks across multiple BSs and make optimal handover decisions based on channel conditions, MEC loads, and user mobility trajectories. This is expected to significantly enhance task offloading success rates, transmission reliability, and overall system performance in high-mobility scenarios.

\vspace{-0.5em}

\footnotesize
\bibliographystyle{IEEEtranN}
\bibliography{IEEEabrv,ref}

\vspace{-0.8cm}

\begin{IEEEbiography}[{\includegraphics[width=1in,height=1.25in,clip,keepaspectratio]{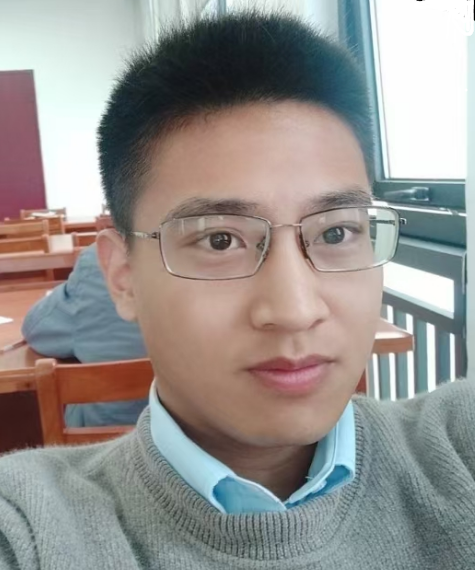}}]{Yuang Chen (Graduate Student Member, IEEE)}
	received the B.S. degree from the Hefei University of Technology (HFUT), Hefei, China, in 2021. He is currently pursuing a doctoral degree in Electronic Engineering and Information Science at the University of Science and Technology of China, Hefei, China. In addition, he is currently working as a full-time research assistant in the department of computing at The Hong Kong Polytechnic University, Hong Kong SAR, China. His research interests include 5G/6G wireless network technology, such as next-generation URLLC, next-generation multiple access technology, wireless network resource allocation and performance optimization, microservice deployment and scheduling, etc.
\end{IEEEbiography}

\vspace{-1cm}

\begin{IEEEbiography}[{\includegraphics[width=1in,height=1.25in,clip,keepaspectratio]{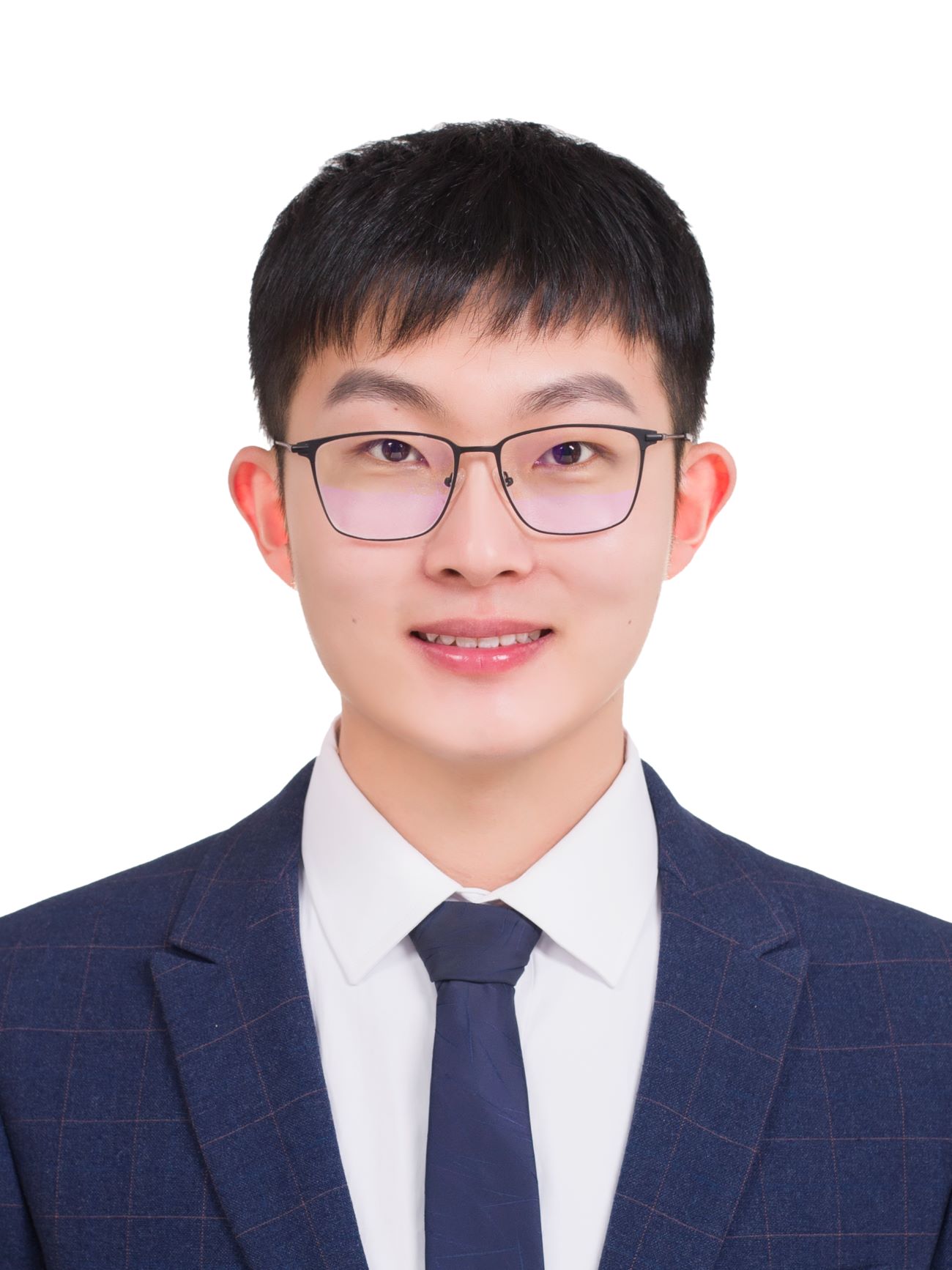}}]{Fengqian Guo}
	received the Ph.D. degree in communication and information systems from the University of Science and Technology of China (USTC), Hefei, China, in 2022. He is currently a associate researcher at the University of Science and Technology of China. His research interests include wireless low-latency transmission and wireless resource optimization.
\end{IEEEbiography}

\vspace{-1cm}

\begin{IEEEbiography}[{\includegraphics[width=1in,height=1.25in,clip,keepaspectratio]{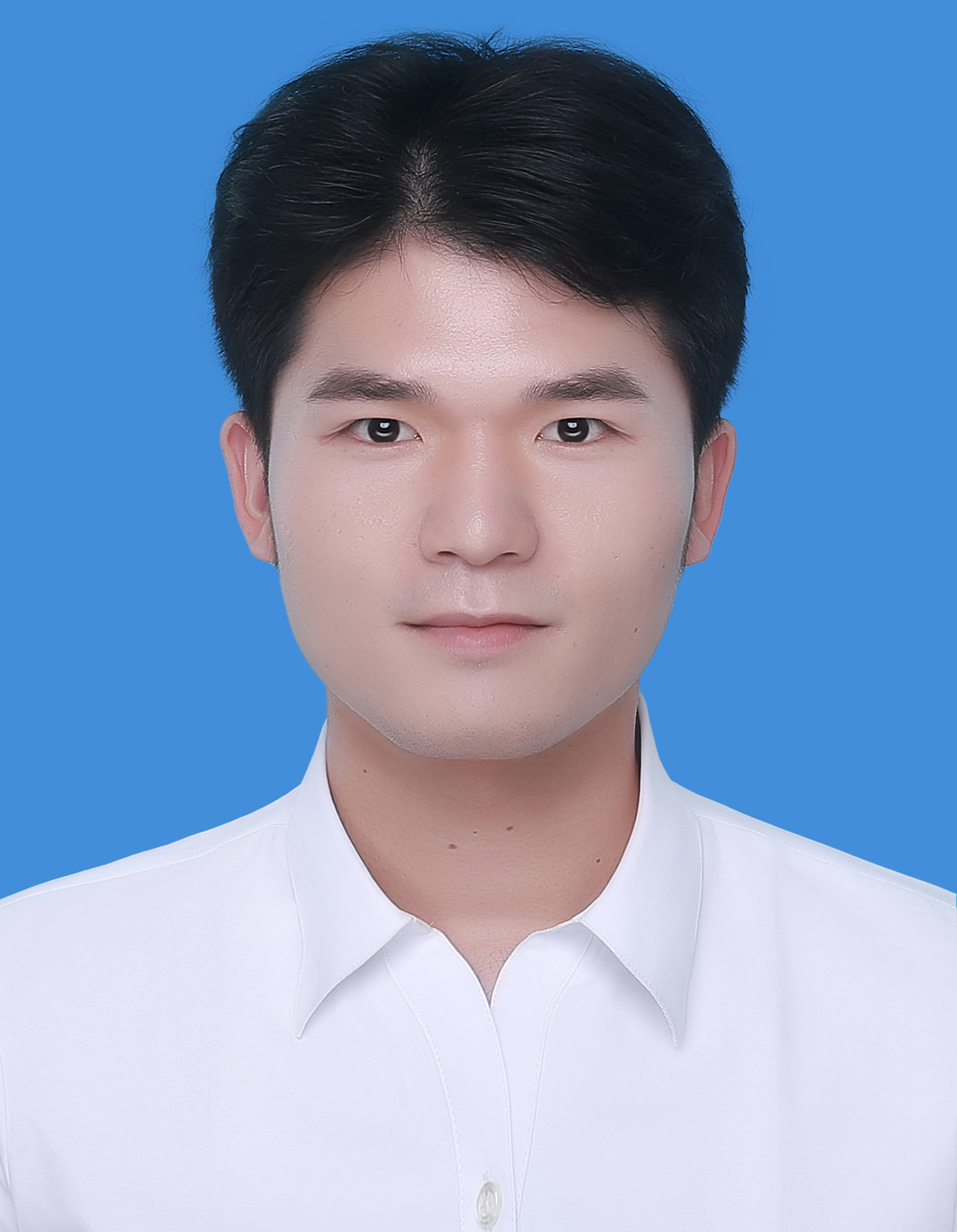}}]{Chang Wu}
	received the B.S. degree from the Dalian Maritime University (DLMU), Dalian, China, in 2021. He is currently working toward the PhD degree in communication and information systems with the Department of Electronic Engineering and Information Science, University of Science and Technology of China (USTC), Hefei, China. His research interests include 5G/6G wireless network technologies, such as architecture, QoS/QoE provision for business transmission and Deep Reinforcement Learning in performance optimization.
\end{IEEEbiography}

\vspace{-2cm}

\begin{IEEEbiography}[{\includegraphics[width=1in,height=1.25in,clip,keepaspectratio]{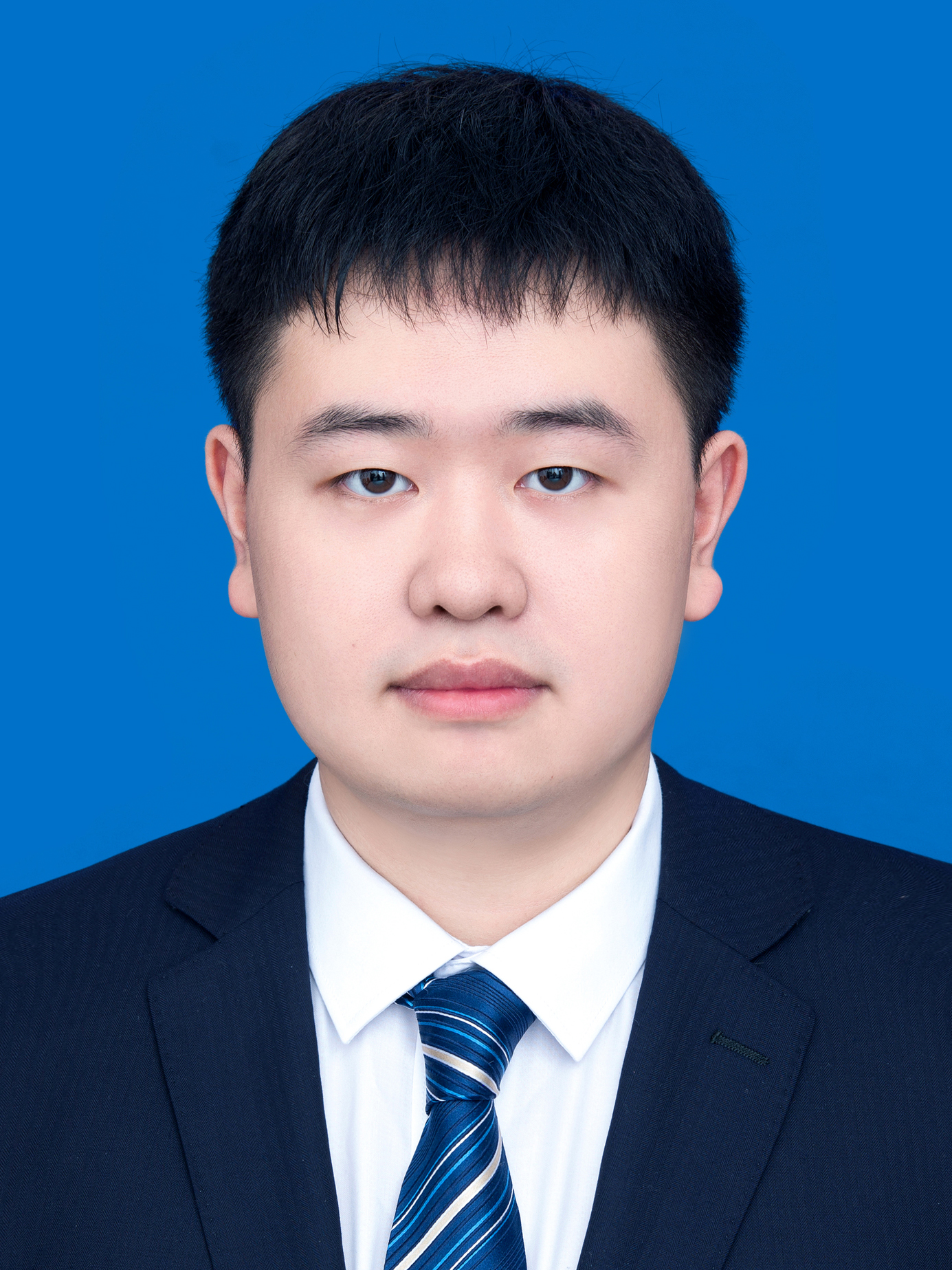}}]{Mingyu Peng}
	received the B.S. degree from the Hefei University of Technology (HFUT), Hefei, China, in 2025. He is currently pursuing the master’s degree with the Department of Electronic Engineering and Information Science, University of Science and Technology of China (USTC), Hefei, China. His research interests include multimedia communication, wireless edge networks, and deep reinforcement learning.
\end{IEEEbiography}

\vspace{-1.3cm}

\begin{IEEEbiography}[{\includegraphics[width=1in,height=1.25in,clip,keepaspectratio]{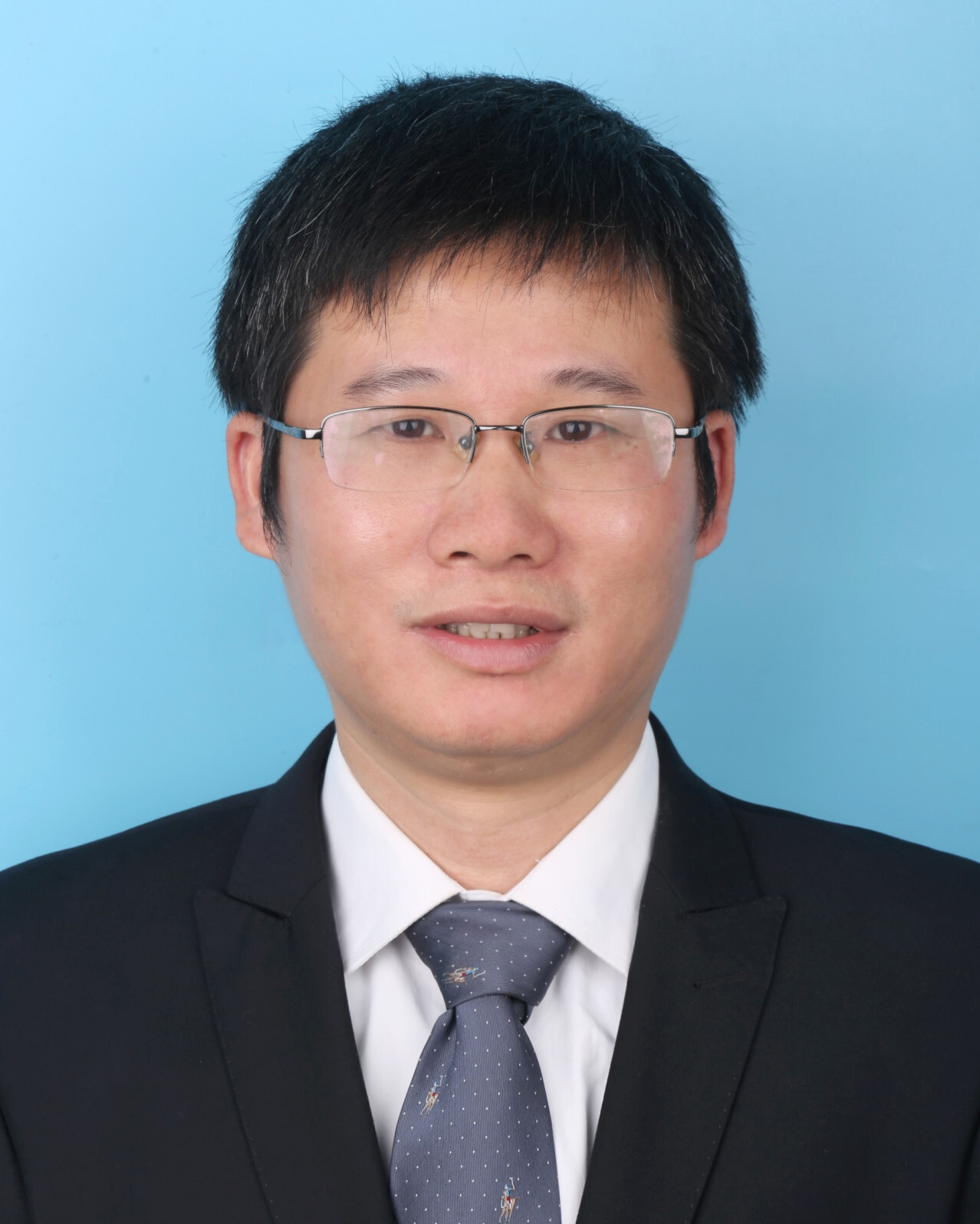}}]{Hancheng Lu (Senior Member, IEEE)} received his Ph.D. in communication and information systems from the University of Science and Technology of China, Hefei, China, in 2005. He is currently a tenured professor in the Department of Electronic Engineering and Information Science at the University of Science and Technology of China. He is also working at the Hefei National Comprehensive Science Center Artificial Intelligence Research Institute, Hefei, China. He has rich research experience in multimedia communication, wireless edge networks, future network architecture and protocols, as well as machine learning algorithms for network communication, involving scheduling, resource management, routing, transmission, and other fields. In the past 5 years, more than 80 papers have been published in top journals such as IEEE Trans and flagship conferences such as IEEE INFOCOM, and have won the Best Paper Award of IEEE GLOBECOM 2021 and the Best Paper Award of WCSP 2019 and WCSP 2016 in the field of communication. In addition, he currently serves as an editorial board member for numerous journals, including the IEEE Internet of Things Journal, China Communications, and IET Communications.
\end{IEEEbiography}

\vspace{-1cm}

\begin{IEEEbiography}[{\includegraphics[width=1.4in,height=1.3in,clip,keepaspectratio]{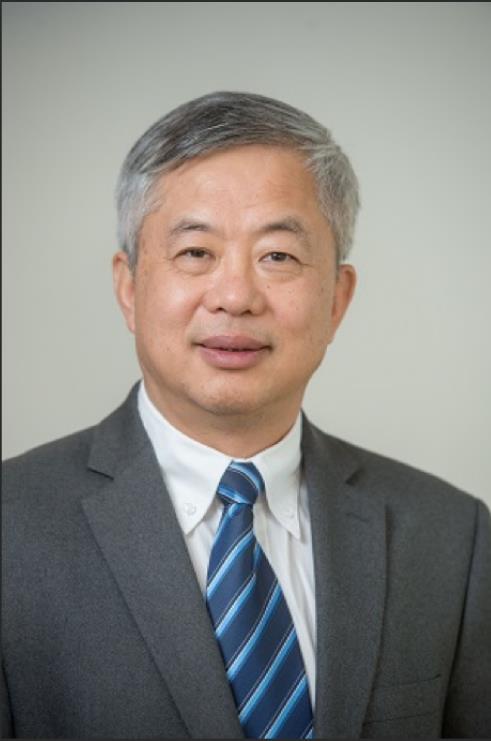}}]{Chang Wen Chen (Life Fellow, IEEE)} received the B.S. degree from the University of Science and Technology of China in 1983, the M.S.E.E. degree from the University of Southern California in 1986, and the Ph.D. degree from the University of Illinois at Urbana–Champaign in 1992. He is currently the Chair Professor in the Department of Visual Computing and the Interim Dean of the School of Computer and Mathematical Sciences at The Hong Kong Polytechnic University. Before his current position, he served as the Dean of the School of Science and Engineering, The Chinese University of Hong Kong, Shenzhen, from 2017 to 2020. He was an Empire Innovation Professor with the University at Buffalo, The State University of New York, from 2008 to 2021. He was an Allen Henry Endow Chair Professor with Florida Institute of Technology from 2003 to 2007. He was a Faculty Member of electrical and computer engineering at the University of Rochester from 1992 to 1996 and at the University of Missouri-Columbia from 1996 to 2003. His research has been funded by both government agencies and industrial corporations. His research interests include multimedia communication, multimedia systems, mobile video streaming, the Internet of Video Things (IoVT), image/video processing, computer vision, deep learning, multimedia signal processing, and immersive mobile video. He was a SPIE Fellow in 2007 and an Elected Member of Academia Europaea in 2021. He and his students have received ten best paper awards or best student paper awards over the past two decades. He received several research and professional achievement awards, such as the Sigma Xi Excellence in Graduate Research Mentoring Award in 2003, the Alexander von Humboldt Research Award in 2010, the University at Buffalo Exceptional Scholar—Sustained Achievement Award in 2012, the SUNY System Chancellor's Award for Excellence in Scholarship and Creative Activities in 2016, the University of Illinois ECE Distinguished Alumni Award in 2019, the Outstanding Overseas Contributor of the China Society of Image and Graphics (CSIS) in China MM 2024, and the SIGMM Outstanding Technical Achievement Award in 2024. He is currently an Associate Editor-in-Chief of IEEE TRANSACTIONS ON BIOMETRICS, BEHAVIOR, AND IDENTITY SCIENCE and a Deputy Editor-in-Chief of the IEEE TRANSACTIONS ON IMAGE PROCESSING. He served as the conference chair for several major IEEE, ACM, and SPIE conferences related to multimedia communications and signal processing. He served as the Editor-in-Chief for IEEE TRANSACTIONS ON MULTIMEDIA from January 2014 to December 2016 and IEEE TRANSACTIONS ON CIRCUITS AND SYSTEMS FOR VIDEO TECHNOLOGY from January 2006 to December 2009. He has been an Editor of several other major IEEE TRANSACTIONS and journals, including PROCEEDINGS OF THE IEEE, IEEE JOURNAL OF SELECTED TOPICS IN SIGNAL PROCESSING, IEEE JOURNAL OF SELECTED AREAS IN COMMUNICATIONS, and IEEE JOURNAL OF SELECTED TOPICS IN SIGNAL PROCESSING.
\end{IEEEbiography}

\end{document}